\documentclass[10pt, letterpaper, twocolumn]{IEEEtran} % a4paper

\usepackage{amsmath}
\usepackage{amssymb}    % mathbb
\usepackage{amsfonts}
\usepackage[english]{babel}
\usepackage{cite} %reference contraction
\usepackage{multirow}
\usepackage{stfloats}    % big equations float

\usepackage{algorithm}   %format of the algorithm
\usepackage{algorithmic} %format of the algorithm

\usepackage{caption} % Figure to Fig.
\usepackage{graphicx}
\usepackage{epsfig}
\usepackage{subfigure} %subfigs
\usepackage{tablefootnote}

\usepackage{diagbox}   %表格内斜线

\usepackage{url}

\usepackage{accents}
\makeatletter
\def\widebar{\accentset{{\cc@style\underline{\mskip10mu}}}}
\def\Widebar{\accentset{{\cc@style\underline{\mskip13mu}}}}
\makeatother

\usepackage{footnote}
\usepackage{pifont}
\usepackage[perpage]{footmisc}  %每页脚注重新编号

\newtheorem{theorem}{Theorem}
\newtheorem{remark}{Remark}
\newtheorem{lemma}{Lemma}
\newtheorem{definition}{Definition}
\newtheorem{proposition}{Proposition}
\newtheorem{corollary}{Corollary}

\usepackage{booktabs}    %表格
\allowdisplaybreaks  %allow equation crossing pages

 \usepackage[draft]{changes}   %审阅模式 % draft or final
	\definechangesauthor[name={Yunquan}, color=orange]{hcs}   %不同修改者的个性显示
	\setremarkmarkup{(#2)}
\begin{document}

\captionsetup[figure]{labelfont={ }, name={Fig.}, labelsep=period} % Figure to Fig.
\pagestyle{empty}

\title{Distributed Sensing with Orthogonal Multiple Access: To code or not to Code?}
\author{Yunquan~Dong,~\IEEEmembership{Member,~IEEE}
\vspace{-8mm}

        \thanks{ Y. Dong is with the School of Electronic and Information Engineering,  Nanjing University of Information Science and Technology, Nanjing 210044, China. Y. Dong is also with the National Mobile Communications Research Laboratory, Southeast University, Nanjing 210096, China (e-mail: yunquandong@nuist.edu.cn).

        This work was supported by the National Natural Science Foundation of China (NSFC) under Grant 61701247, the Startup Foundation for Introducing Talent of NUIST under Grant 2243141701008, the open research fund of National Mobile Communications Research Laboratory, Southeast University, under grant No. 2020D09.
        %The authors would like to thank the editor and anonymous reviewers for their valuable comments that improve the presentation of the paper.
        }
        }

%\author{\authorblockN{Yunquan Dong  \\ %\authorrefmark{2} \\
%\authorblockA{ \normalsize
%%\authorrefmark{1}
%                                 School of Electronic and Information Engineering, \\
%                                 Nanjing University of Information Science \& Technology, Nanjing, China \\
%yunquandong@nuist.edu.cn
%}
%}
%}

%\date{}
\maketitle
\thispagestyle{empty}

\begin{abstract}
We consider the estimation distortion of a distributed sensing system with finite number of sensor nodes, in which
    the nodes observe a common phenomenon and transmit their observations to a fusion center over orthogonal channels.
In particular, we investigate whether the coded scheme (separate source-channel coding) outperforms the uncoded scheme (joint source-channel coding) or not.
    To this end, we explicitly derive the estimation distortion of a coded heterogeneous sensing system with diverse node and channel configurations.
Based on this result, we  show that in a homogeneous sensing system with identical node and channel configurations,  the coded scheme outperforms the uncoded scheme if the number of nodes is $K=1$ or $K=2$.
    For homogenous sensing systems with $K\geq3$ nodes and general heterogeneous sensing systems, we also present explicit conditions for the coded scheme to perform better than the uncoded scheme.
Furthermore, we propose to minimize the estimation distortion of heterogeneous sensing systems with hybrid coding, i.e., some nodes use the coded scheme and other nodes use the uncoded scheme.
    To determine the optimal hybrid coding policy, we  develop three greedy algorithms, in which the pure greedy algorithm minimizes distortion greedily, the group greedy algorithm improves performance by using a group of potential sub-polices, and the sorted greedy algorithm reduces computational complexity by using a pre-solved iteration order.
Our numerical and Monte Carlo results show that the proposed algorithms closely approach the optimal policy in terms average estimation distortion.
\end{abstract}


\section{Introduction}
Sensing the phenomenon of interest with spatially distributed sensors and estimating the phenomenon with a centralized fusion center is a fundamental application of wireless sensor networks (WSNs).
    This scenario is referred to as \textit{distributed sensing} and has wide applications in environment monitoring \cite{air-monitring}, industrial automation \cite{industrial-auto}, target surveillance and tracking \cite{uav-tracking}, and so on.
 In such sensing systems, one well known problem is whether the separate source-channel coding outperforms the joint source-channel coding or not.

For point-to-point channels, Shannon has proved that the separate source-channel coding is an optimal strategy \cite{sampling_theorem}.
    That is, we can consider the problem of source coding separately from the problem of channel coding and transmit a stationary ergodic source over a channel if and only if the entropy rate (or the rate-distortion function) of the source is less than the capacity of the channel.
In the network context, however, the separate source-channel coding is not optimal in general \cite{Gastpar-tocode-2003, Gastpar-conf-2003, Gastpar-uncode-opt-2008}.
    As shown in \cite{Gastpar-tocode-2003}, the joint source-channel coding strictly outperforms the separate channel-source coding in distributed sensing systems with coherent Gaussian multiple access.
In particular, the joint source-channel coding can be implemented through uncoded forwarding, i.e., the observations are amplified and forwarded to the fusion center without coding \cite{Gastpar-conf-2003}.
    Later, a multi-letter rate-distortion region was developed for distributed sensing systems with separate source-channel coding and orthogonal multiple access  \cite{Jinjun-oMAC-2007}.
For a two-node system, it was further shown that the joint source-channel coding is strictly inferior to the separate source-channel coding in terms of rate-distortion and power-distortion regions.

In this paper, we consider the estimation distortion of a distributed sensing system with orthogonal multiple access and finite number of nodes.
    To be specific, the observations made by the nodes are corrupted by Gaussian noises and are transmitted to the fusion center through orthogonal Gaussian channels.
With the received observations, the fusion center estimates the source signal using a best-linear-unbiased-estimator (BLUE) \cite{estimation_book}.
    Under the mean-squared error (MSE) measure, we evaluate the estimation distortion of the system with the following easy-implementing coding schemes:
     \begin{itemize}
       \item the \textit{coded scheme:} each node encodes its distorted observations with lossy source coding and then transmits the obtained index to the fusion center using ideal channel coding, i.e., employing the \textit{separate source-channel coding};
       \item the \textit{uncoded scheme:} each node amplifies its observations and then transmits them to the fusion center directly, i.e., using the \textit{joint source-channel coding};
       \item the \textit{hybrid coding:} some nodes use the coded scheme and the other nodes use the uncoded scheme.
     \end{itemize}

    Moreover, we consider the following two types of sensing systems:
    \begin{itemize}
      \item the \textit{heterogeneous sensing system} in which the nodes have different transmit powers and different observation noise powers while the channels have its own channel gains;
      \item the \textit{homogeneous sensing system} in which the nodes have the same observation noise power and transmit power while all the channels share the same channel gain and channel noise power.
    \end{itemize}

Furthermore, we propose three greedy algorithms to determine the optimal coding policy for heterogeneous sensing systems with hybrid coding.
    Namely, the \textit{pure greedy} algorithm which adds a local optimal node to the active node set in each iteration, the \textit{group greedy} algorithm which iteratively updates the coding policy based on a group of potential policies obtained in the previous iteration, and the \textit{sorted greedy} algorithm which determines the coding policy of each node according a pre-solved order.

\subsection{Main Contributions}
The main contributions of the paper are as follows.
\begin{enumerate}
  \item [1)] \textit{Fundamental limits:} We explicitly present the estimation distortion of heterogeneous sensing systems and homogeneous sensing systems with finite number of nodes, with the coded scheme or/and the uncoded scheme, respectively.
  \item [2)] \textit{Necessary and sufficient conditions:} We explicitly present the condition for the coded scheme to be optimal in homogeneous sensing systems.
        We show that when the number of nodes is small (e.g., $K=1, 2$), the coded scheme outperforms the uncoded scheme regardless of the observation signal-to-noise ratio (SNR) and the channel SNR of the system;
            for $K\geq3$, we show that the coded scheme is optimal in the low observation SNR regime and derive the corresponding boundary explicitly.
        For heterogeneous sensing systems, we also present a general condition for the coded scheme to be optimal.
  \item [3)] \textit{Efficient algorithms:} We propose three greedy algorithms to search the near optimal coding policy for heterogeneous sensing systems with hybrid coding.
        While the group greedy algorithm performs the best in terms of estimation distortion and policy error rate, the sorted greedy algorithm find its solution with the minimum computational complexity, and the pure greedy algorithm achieves a good trade-off between distortion performance and computational complexity.
\end{enumerate}

\subsection{Related Works}
There have been many works investigating how the multiple access scheme affects the estimation distortion of distributed sensing systems.
    In \cite{Gastpar-tocode-2003, Gastpar-conf-2003, Gastpar-uncode-opt-2008}, the nodes communicate with the fusion center through a coherent multiple access channel (MAC), i.e., all the nodes are synchronized at subcarrier-level and their transmitted signals are mixed naturally at the fusion center.
 In this case, it is concluded that the uncoded scheme outperforms the separate source-channel coding, and has the optimal asymptotic scaling behavior when the source signal is Gaussian distributed \cite{Gastpar-tocode-2003, Gastpar-conf-2003, Gastpar-uncode-opt-2008}, or satisfy a certain mean condition \cite{Liu-SSP-2005}.
    It is also noted that the coherent MAC shares a similar concept with the non-orthogonal multiple-access (NOMA), which has been extensively investigated in recent years \cite{Ding-Tutorial-2018, Ding-JSAC-2017}.
In particular, NOMA is a very promising technology in terms of throughput and latency, since many signaling overheads (e.g., require-to-send and clear-to-send) can be eliminated or reduced.
    Facilitated with the low-complexity multi-user detection technology,  therefore, NOMA will be a key enabler of the fifth generation (5G) cellular communication system, in which the limited time, frequency, and power resources need to be shared among a large number devices.
        In coherent MAC and NOMA based systems, however, the requirement on the subcarrier-level synchronization among sensor nodes presents serious challenge to their implementations.
When the full coordination among sensor nodes is unavailable, therefore, it is more convenient to communicate with the fusion center through orthogonal multiple access, such as time/frequency/code-division-multiple-access (TDMA/FDMA/CDMA) \cite{Jinjun-oMAC-2007, Xiao-Estimation-2006}.
%Generally speaking, coherent MAC yields significantly superior estimation performance over the orthogonal multiple access \cite{Liu-Type-2007}.
    Later, a hybrid multiple-access scheme was proposed in \cite{Chung-TVT-2011}, which allows the nodes of a same group to share the same channel and the groups to use orthogonal channels.

The power allocation of distributed sensing systems also attracts many attentions  \cite{Chung-TVT-2011, Xiao-Estimation-2006, Xiao-Linear-2008,  Cui-Estimation-2007}.
    In \cite{Xiao-Estimation-2006},  the optimal power scheduling and the optimal quantization scheme were considered for a  sensing system with orthogonal MAC and decentralized quantization.
The authors further showed that the optimal power scheduling can be solved through convex optimizations and admits a  distributed
implementation if the MAC is coherent \cite{Xiao-Linear-2008}.
    The estimation outage probability, the estimation diversity, and the optimal power allocation were investigated for an uncoded sensing system with orthogonal MAC in \cite{Cui-Estimation-2007}.
From a node collaboration perspective, energy efficient message sharing (among neighboring nodes) schemes and node selection schemes were proposed in \cite{Sijia-collab-2016} and \cite{Sijia-selec-2016}, respectively.
    It was also shown in \cite{Fekri-WCNC-2005} and \cite{Fekri-SECON-2004} that by using the low density parity-check (LDPC) codes in distributed source coding and joint source-channel coding, the achievable rate region and energy efficiency can closely approach the Slepian-Wolf limit.
The power scheduling of distributed sensing systems has also been exhaustedly studied in scenarios of energy harvesting and wireless power transfer, in which the sensor nodes harvest energy from the ambient environment or from a radio energy source.
    Since the energy arrival processes are random, the instantaneous distortions of these systems are also random, and thus we can minimization the average MSE distortion of the system by using the Markov decision process (MDP) based stochastic control \cite{Dey-MDP-2015}, the Lyapunov optimization technique \cite{Zhou-TVT-2016}, or the semidefinite relaxation framework \cite{Mai-relex-2017}.

\subsection{Organizations}

This paper is organized as follows.
   Section~\ref{sec:2_model} presents our system model, coding model, and estimation model.
 We investigate the estimation distortions of heterogeneous sensing systems and homogeneous sensing systems in Section~\ref{sec:3_hetero} and Section  \ref{sec:4_homo}, respectively.
    In Section \ref{sec:4_homo}, we also present the conditions for the coded scheme to be optimal.
In Section \ref{sec:5_hybrid}, we explicitly present the condition for the coded scheme to be optimal in heterogeneous sensing systems.
    Moreover, we present the estimation distortion of heterogeneous sensing systems with hybrid coding and propose three efficient  algorithms to solve the near optimal coding policy.
    The obtained results are presented via numerical results and verified via Monte Carlo simulations in Section~\ref{sec:6_simulation}.
Finally,  we shall conclude our work in Section~\ref{sec:5_conclusion}.

\textit{Notations}: Boldface letters indicate vectors and matrices. $\boldsymbol{v}^{\text{T}}$ denotes the transpose of the vector.
    $\textbf{1}_K$ represents a $K$ dimensional vector of ones and $\textbf{I}_K$ represents a $K\times K$ dimensional unit matrix.
$\text{diag}(c_1,c_2,\cdots,c_K)$ is a $K\times K$ dimensional diagonal matrix.
    $\mathcal{K}=\{1,2,\cdots,|\mathcal{K}|\}$ is a set of integers and $|\mathcal{K}|$ denotes the cardinality of the set.
For any subset $\mathcal{A}\in \mathcal{K}$, we denote $\boldsymbol{v}_{-\mathcal{A}}=\{v_i: i\in \mathcal{K-A}\}$.
    $x\perp y|z$ means that random variables $x$ and $y$ are independent from each other when $z$ is given.
$\oplus$ is the bitwise modulo-two sum operator and $n\choose k $ is the Binomial coefficient.

%
% The set function $f:2^E\rightarrow \mathcal{R}^+$ is a mapping from all subsets of $E$ (there are a total of $2^{|E|}$ subsets) to the positive real numbers.
%A permutation of the set $E$ is denoted by $\pi$ and  $\pi(k)$, $1\leq k\leq |E|$, represents the element of the
%set located in the $k$th position after the permutation.
%For an $a$-dimensional vector $\boldsymbol{x}=\{x_1,x_2\cdots,x_a\}\in \mathcal{R}^a$ and $S\subseteq E $, $\boldsymbol{x}(S)$ denotes $\sum_{k\in S}x_k$.

\begin{figure}[!t]
\centering
\includegraphics[width=3.2in]{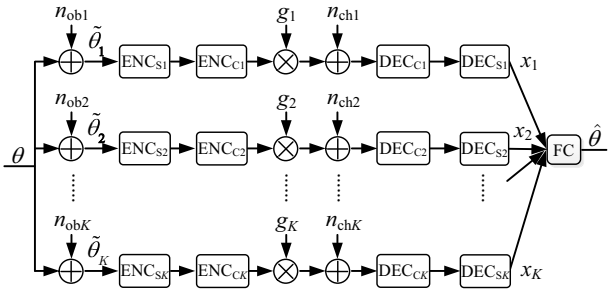}
\caption{Distributed sensing model. $\text{ENC}_\text{S}$ denotes the source encoder,  $\text{ENC}_\text{C}$ denotes the channel encoder, $\text{DEC}_\text{S}$ denotes the source decoder, $\text{DEC}_\text{C}$ denotes the channel decoder, and $\text{FC}$ is the fusion center.} \label{fig:net_model}
\end{figure}

\section{System Model}\label{sec:2_model}
Consider a heterogeneous sensor network of $K$ nodes, each node observing a common phenomenon (source signal) $\theta$ characterized by a Gaussian process.
    The observations will be separately encoded and transmitted to a remote fusion center over orthogonal Gaussian channels, as shown in Fig. \ref{fig:net_model}.
Afterwards, the fusion center will decode the received signals and make an estimation of $\theta$.

\subsection{Source Model and Channel Model}
We assume that the source signal $\{\theta_n, n= 1,2,\cdots\}$ is a  Gaussian process with zero-mean and variance $\sigma^2_\theta$.
    We denote the node set as $\mathcal{K}=\{1,2,\cdots, K\}$ and the frequency bandwidth of the transmitted signals as $W$.
 We further denote the transmit power of node $k$ as $P_k$, the power gain of the $k$-th channel as $g_k$, and the corresponding channel noise power as $\sigma^2_{\text{ch}k}$.
    Thus, the channel SNR of the $k$-th channel can be expressed as $\gamma_{\text{ch}k}=\frac{g_kP_k}{\sigma^2_{\text{ch}k}}$.

On the source model and the channel model, we consider the following assumptions.
\begin{itemize}
  \item [A1] $\{\theta_n\}$ is stationary, as is assumed in \cite{Gastpar-tocode-2003, Gastpar-conf-2003, Gastpar-uncode-opt-2008, Jinjun-oMAC-2007, Cui-Estimation-2007, Sijia-selec-2016, Sijia-collab-2016}.
 According to Rhhlin's ergodic decomposition theorem, we know that any stationary source having an alphabet with a separable $\sigma$-algebra can be considered as a unique mixture of a (possibly uncountable) number of stationary ergodic sub-sources.
     This means that there exists a universal and optimal code/estimation as if we had known in advance which sub-source was chosen \cite{Gray-TIT-1974}.
 For most practical scenarios, therefore, $\{\theta_n\}$ can be treated as an ergodic source without the assumption of ergodicity.
     Moreover, the obtained results can be generalized to systems with non-stationary sources  \cite{Reily-ComST-2014} or quasi-stationary sources \cite{Joda-TCom-2013}, as long as the period before the next change in distribution is sufficiently long for the required source coding and channel coding (e.g., we have $10^5$ channel uses per second with block length $T_{\text{B}}=10~\mu$s).
  \item [A2] The channel gain of each link is available at the fusion center.
                        With the strong computational capability of the fusion center, this can be realized by efficient channel estimation techniques, e.g., \cite{chanelesti-TSP-2010,chanelesti-TSP-2018}.
  \item [A3] Pairwise synchronization between the nodes and the fusion center.
  \item [A4] $g_k=1$ for all $k\in\mathcal{K}$.
                        It should be noted that by varying the channel noise power, we can evaluate the influence of channel attenuation equivalently.
                            For a fixed noise power $\sigma^2_{\text{n}k}$,  for example, we can set $\sigma^2_{\text{ch}k}=\sigma^2_{\text{n}k} d_k^\alpha$ and change $\sigma^2_{\text{ch}k}$ instead of varying distance $d_k$, in which $\alpha$ is the pathloss exponent.
                    By considering some additional multiplicative gains, this model can also be extended to systems with block fading (cf. Subsection \ref{subcec:fading_sim}).
  \item [A5] Each node observes the phenomena at a rate $2W$.
                        In each second, therefore, the number of observations is equal to the number of channel uses for each node.
\end{itemize}

    In particular, we do not need the phase shifts of channel gains and the subcarrier-level synchronization at the fusion center, as in \cite{Cui-Estimation-2007}.

 Based on assumption A4,  the channel SNR of the $k$-th channel can  be rewritten as $\gamma_{\text{ch}k}=\frac{P_k}{\sigma^2_{\text{ch}k}}$.

\subsection{Coding Model}
In a certain epoch, we denote the source sample  as ${\theta}$ and the observation of sensor node $k$ as  $\widetilde{{\theta}}_k$.
    Due to accuracy issues, the observation suffers from independently and identically distributed (i.i.d.) Gaussian observation noise ${ n}_{\text{ob}k}$ with zero-mean and variance $\sigma^2_{\text{ob}k}$.
We then have $\widetilde{{\theta}}_k={\theta}+{n}_{\text{ob}k}$ and denote the observation SNR as $\gamma_{\text{o}k}=\frac{\sigma^2_\theta}{\sigma^2_{\text{ob}k}}$.

    We assume that the distorted observations are encoded with separate source-channel coding at each node distributedly.
To be specific, each node encode a sequence of its observations into a message index using ideal lossy source coding \cite[Chap.~10.2, \textit{Definition} 10.7]{Cover-IT-Book}.
    The index is then encoded into an ideal channel codeword, which would be transmitted to the fusion center through a separate Gaussian channel (e.g., in a separate frequency band).
When the fusion center receives a distorted channel codeword, it decodes the index accurately and then obtains a recovery ${x}_k$ of $\widetilde{{\theta}}_k$ with a certain quantization distortion  $\sigma^2_{\text{qu}k}$, i.e., $\mathbb{E}[(\widetilde{\theta}-x_{k})^2]=\sigma^2_{\text{qu}k}$.
    When the codewords of source coding and channel coding are sufficiently long,  the rate-distortion limit of each node and the capacity of each channel can be approached, which are, respectively, given by~\cite[Chap.~9.1, \textit{Theorem} 9.1.1 and Chap.~10.2, \textit{Theorem} 10.3.2]{Cover-IT-Book}
\begin{align}
    r_{\text{ch}k} & = W \log\left( 1+\gamma_{\text{ch}k} \right), \\
\label{eq:rate}
  r_{\text{sc}k} & = W \log \frac{\sigma^2_\theta+\sigma^2_{\text{ob}k}}{\sigma^2_{\text{qu}k}}.
\end{align}
Thus, the quantization distortion $\sigma^2_{\text{qu}k}$ can be expressed as
\begin{equation} \label{rt:sigma_qu}
    \sigma^2_{\text{qu}k} = \frac{ \sigma^2_\theta+\sigma^2_{\text{ob}k} }  { 1+\gamma_{\text{ch}k} }.
\end{equation}

As shown in~\cite[Chap.~10.3, Fig. 10.5]{Cover-IT-Book}, the rate-distortion limit approaching source coding from observation $\widetilde{\theta}_{k}$ to recovery $x_{k}$ can be characterized by the following test channel
    \begin{equation} \label{eq:x_first}
        \widetilde{\theta}_k = x_k + n_{\text{qu}k},
    \end{equation}
    in which $n_{\text{qu}k}$ is the i.i.d. Gaussian quantization noise with zero mean and variance $\sigma^2_{\text{qu}k}$, and $x_k$ is a Gaussian random variable generated according to a certain optimal distribution $p(x_k|\widetilde{\theta}_k)$.
Moreover, $n_{\text{qu}k}$ and $x_k$ are independent from each other.
     For $k\neq j, k,j,\in\mathcal{K}$, however, it should be noted that $n_{\text{qu}k}$ and $n_{\text{qu}j}$ are correlated with each other, since both $n_{\text{qu}k}$ and $n_{\text{qu}j}$ are correlated with $\theta$.
 For any given $\theta$, therefore, the joint probability density function of $n_{\text{qu}j}$ and $n_{\text{qu}k}$ should be expressed as
\begin{align}\label{dr:cond_independ}
  p(n_{\text{qu}k}, n_{\text{qu}j}|\theta) &= p(n_{\text{qu}k}|\theta) p(n_{\text{qu}j}|\theta, n_{\text{qu}k}).
\end{align}

However, $n_{\text{qu}k}$ can actually be removed from the condition in \eqref{dr:cond_independ}, since $n_{\text{qu}k}$ provides no more information about $n_{\text{qu}j}$ than $\theta$.
    That is, $n_{\text{qu}k}\rightarrow\theta\rightarrow n_{\text{qu}j}$ forms a Markov chain and $n_{\text{qu}j}$ does not depend on $n_{\text{qu}k}$ when $\theta$ is given.
    Thus, \eqref{dr:cond_independ} can be rewritten as
\begin{align}\label{rt:cond_independ}
  p(n_{\text{qu}k}, n_{\text{qu}j}|\theta) = p(n_{\text{qu}k}|\theta) p(n_{\text{qu}j}|\theta),
\end{align}
which means that $n_{\text{qu}k}$ and $n_{\text{qu}j}$ are conditionally independent from each other when $\theta$ is given.

    Since $\widetilde{\theta}_k= \theta+n_{\text{ob}k}$ is a noisy version of $\theta_k$, the recovery vector $\boldsymbol{x}$ (cf. \eqref{eq:x_first}) of $\theta$ can further be expressed as
    \begin{equation}\label{eq:model}
      \boldsymbol{x} = \theta\textbf{1}_K + \boldsymbol{n}_\text{ob} - \boldsymbol{n}_\text{qu},
    \end{equation}
    in which
    \begin{align}\label{df:x_n}
      \boldsymbol{x} &= [x_1, x_2,\cdots, x_K], \\
      \boldsymbol{n}_{\text{ob}} &= [n_{\text{ob}1}, n_{\text{ob}2},\cdots, n_{\text{ob}K}], \\
      \label{df:n_quk}
      \boldsymbol{n}_{\text{qu}} &= [n_{\text{qu}1}, n_{\text{qu}2},\cdots, n_{\text{qu}K}].
    \end{align}

\subsection{Estimation Model}
To estimate the source signal $\theta$ with recovery vector $\boldsymbol{x}$, a best-linear-unbiased-estimator is used at the fusion center~\cite[Lesson 9]{estimation_book}.
    To be specific, the fusion center estimates $\theta$ in each slot through
\begin{equation}\label{df:blue_equation}
    \hat{\theta} = \boldsymbol{f}^\text{T} \boldsymbol{x},
\end{equation}
in which $\boldsymbol{f}=[f_1,f_2,\cdots,f_K]$ is a positive weighting vector satisfying $\boldsymbol{f}^\text{T}\textbf{1}_K=1$.
    By optimizing the MSE of the estimator over all possible weighting vectors and using Assumption A4, the minimum achievable MSE is given by~\cite[Lesson 9, \textit{Theorem} 9.3]{estimation_book}
\begin{equation}\label{eq:mmse_blue}
    D = \min\limits_{\boldsymbol{f}} \mathbb{E}[(\hat{\theta}-\theta)^2] = (\textbf{1}_K^\text{T} \boldsymbol{\Sigma}_{\boldsymbol{n}}^{-1} \textbf{1}_K)^{-1},
\end{equation}
    in which $\boldsymbol{\Sigma}_{\boldsymbol{n}}$ is the covariance matrix of the following vector of total noise
    \begin{equation}
        \boldsymbol{n}=\boldsymbol{n}_{\text{qu}} - \boldsymbol{n}_{\text{ob}}.
    \end{equation}

\section{Distortion of Coded Heterogeneous Sensing Systems} \label{sec:3_hetero}
In this section, we shall  investigate the correlation structure among quantization noises and  present the estimation distortion of coded heterogeneous systems explicitly.

\subsection{Correlation Among Quantization Noises}
We denote the following $(K+1)$ dimensional augmented noise vector with zero-mean and covariance $\boldsymbol{\Sigma}_{\widetilde{\boldsymbol{n}}}$ as,
\begin{equation}
\widetilde{\boldsymbol{n}}=[\boldsymbol{n},\theta]. %= [n_{\text{qu}1}-n_{\text{ob}1}, n_{\text{qu}2}-n_{\text{ob}2},\cdots, n_{\text{qu}K}-n_{\text{ob}K},\theta].
\end{equation}

We rewrite $\boldsymbol{\Sigma}_{\widetilde{\boldsymbol{n}}}$ as
\begin{equation}
    \boldsymbol{\Sigma}_{\widetilde{\boldsymbol{n}}} =
    \left[
      \begin{array}{cc}   %该矩阵一共3列，每一列都居中放置
        \boldsymbol{\Sigma}_{11} & \boldsymbol{\Sigma}_{12} \\  %第一行元素
        \boldsymbol{\Sigma}_{21} & \boldsymbol{\Sigma}_{22} \\  %第二行元素
      \end{array}
    \right],
\end{equation}
in which $\boldsymbol{\Sigma}_{11}$ is a $K\times K$ dimensional matrix composed by the elements in the first $K$ rows and the first $K$ columns of $\boldsymbol{\Sigma}_{\widetilde{\boldsymbol{n}}}$.
    Thus,  $\boldsymbol{\Sigma}_{11}$ exactly equals the covariance matrix of $\boldsymbol{n}$, i.e.,
    \begin{equation}\label{dr:sigma_n_11}
        \boldsymbol{\Sigma}_{11}=\boldsymbol{\Sigma}_{\boldsymbol{n}}.
    \end{equation}

We denote the precision matrices of $\widetilde{\boldsymbol{n}}$ and $\boldsymbol{n}$, respectively, as $\textbf{Q}_{\widetilde{\boldsymbol{n}}}$ and $\textbf{Q}_{\boldsymbol{n}}$.
    That is, $\textbf{Q}_{\widetilde{\boldsymbol{n}}}= \boldsymbol{\Sigma}_{\widetilde{\boldsymbol{n}}}^{-1}$ and $\textbf{Q}_{\boldsymbol{n}} = \boldsymbol{\Sigma}_{\boldsymbol{n}}^{-1}$.

Before solving the covariance matrix $\boldsymbol{\Sigma}_{\boldsymbol{n}}$, we present a useful lemma first \cite[Chap.~2.2, \textit{Theorem} 2]{Random-Field-2005}.
\begin{lemma}\label{lem:q_cond_0}
    Let $\boldsymbol{n}$ be a Gaussian distributed vector with precision matrix $\textbf{Q}>0$.
        For any $i\neq j$, we have
    \begin{equation}
        n_i \perp n_j | \boldsymbol{n}_{-ij} \quad \Leftrightarrow \quad \textbf{Q}_{kj}=0,
    \end{equation}
    in which $\perp$ stands for the statistical independence between random variables and $-ij=\mathcal{K}-\{i,j\}$.
\end{lemma}

Since $n_{\text{qu}k}$ and $n_{\text{qu}j}$ are independent from each other when $\theta$ is given (cf. \eqref{rt:cond_independ}), it is clear that $(\textbf{Q}_{\widetilde{\boldsymbol{n}}})_{kj}=0$ for all $k\neq j$ and $k, j\in \mathcal{K}$.

Based on \textit{Lemma} \ref{lem:q_cond_0}, the covariance matrix $\boldsymbol{\Sigma}_{\boldsymbol{n}}$  can be obtained readily, as shown in the following proposition.
\begin{proposition} \label{prop:sigma_n}
    The covariance matrix $\boldsymbol{\Sigma}_{\boldsymbol{n}}$ of total noise $\boldsymbol{n}$ is determined by
\begin{align}\label{rt:sigma_n}
      (\boldsymbol{\Sigma}_{\boldsymbol{n}})_{kk} &= \sigma^2_{\text{ob}k} + \sigma^2_{\text{qu}k} -
                    \frac{2\sigma^2_{\text{ob}k} \sigma^2_{\text{qu}k}} {\sigma^2_\theta+\sigma^2_{\text{ob}k}}, \\
    (\boldsymbol{\Sigma}_{\boldsymbol{n}})_{kj} &=
                    \frac{\sigma^2_\theta\sigma^2_{\text{qu}k} \sigma^2_{\text{qu}j}} {(\sigma^2_\theta+\sigma^2_{\text{ob}k})(\sigma^2_\theta+\sigma^2_{\text{ob}j})},
\end{align}
for all $k\neq j,~k, j\in \mathcal{K}$.
\end{proposition}

\begin{proof}
    First, we note that $(\boldsymbol{\Sigma}_{\boldsymbol{n}})_{kk} = \mathbb{E}[(n_{\text{qu}k}-n_{\text{ob}k})^2]$ is the total noise power of node $k$,  $(\boldsymbol{\Sigma}_{\widetilde{\boldsymbol{n}}})_{K+1,K+1}=\sigma^2_\theta$ is the source signal power, and $(\boldsymbol{\Sigma}_{\widetilde{\boldsymbol{n}}})_{k,K+1}=(\boldsymbol{\Sigma}_{\widetilde{\boldsymbol{n}}})_{K+1,k}$ is the correlation between $\theta$ and the total noise of node $k$.
        Since $\boldsymbol{\Sigma}_{11}=\boldsymbol{\Sigma}_{\boldsymbol{n}}$ (cf. \eqref{dr:sigma_n_11}) and correlation $(\boldsymbol{\Sigma}_{\widetilde{\boldsymbol{n}}})_{k,K+1}$ can be calculated readily, the remaining unknown elements in $\boldsymbol{\Sigma}_{\widetilde{\boldsymbol{n}}}$ are $(\boldsymbol{\Sigma}_{\widetilde{\boldsymbol{n}}})_{kj}$ for all $k\neq j,  k,j\in \mathcal{K}$.
    Second, according to \textit{Lemma} \ref{lem:q_cond_0}, we have $(\textbf{Q}_{\widetilde{\boldsymbol{n}}})_{kj}=0$ for all $k\neq j, k,j\in \mathcal{K}$.
        Thus, the remaining unknown elements of $\textbf{Q}_{\widetilde{\boldsymbol{n}}}$ include $(\textbf{Q}_{\widetilde{\boldsymbol{n}}})_{kk}$, $(\textbf{Q}_{\widetilde{\boldsymbol{n}}})_{k,K+1}$, and $(\textbf{Q}_{\widetilde{\boldsymbol{n}}})_{K+1,k}$ for all $k\in \mathcal{K}$.
    Fortunately, all these unknown variables can be solved through equation $\boldsymbol{\Sigma}_{\widetilde{\boldsymbol{n}}}^\text{T} \textbf{Q}_{\widetilde{\boldsymbol{n}}}=\textbf{I}_{K+1}$.
    More details are shown in \textit{Appendix}~\ref{prf:prop_sigma_n}.
\end{proof}

\subsection{Minimum Distortion}
Based on \textit{Proposition} \ref{prop:sigma_n} and equation \eqref{eq:mmse_blue}, the minimum achievable distortion of  the coded heterogeneous sensing system can be obtained, as shown in the following theorem.

\begin{theorem}\label{th:distortion}
    The minimum achievable distortion of the coded heterogeneous sensing system is given by
    \begin{equation}\label{rt:distortion}
        D = \sigma^2_\theta\left(\sum_{k=1}^K \frac{1}{\lambda_k} - \frac{1}{1+  \sum_{k=1}^K \frac{u_k^2}{ \lambda_k} }
                \left(\sum_{k=1}^K\frac{u_k}{\lambda_k} \right)^2\right)^{-1},
    \end{equation}
    in which
    \begin{align}\label{rt:u_k}
        u_k &= \frac{1}{1+\gamma_{\text{ch}k}}, \\
        \label{rt:lamda_k}
        \lambda_k & =
                        \frac{(1+\gamma_{\text{ch}k} + \gamma_{\text{ob}k})\gamma_{\text{ch}k}}
                  {(1+\gamma_{\text{ch}k})^2 \gamma_{\text{ob}k} }.
    \end{align}
\end{theorem}

\begin{proof}
    To prove \textit{Theorem} \ref{th:distortion}, we need the inverse matrix of covariance matrix $\boldsymbol{\Sigma}_{\boldsymbol{n}}$, which can be obtained using equation $(\textbf{A}+b\boldsymbol{u} \boldsymbol{v}^{\text{T}} )^{-1} = \textbf{A}^{-1} - \frac{b}{1+ b \boldsymbol{v}^{\text{T}}  \textbf{A}^{-1}\boldsymbol{u} } \textbf{A}^{-1} \boldsymbol{u} \boldsymbol{v}^{\text{T}} \textbf{A}^{-1}$~\cite[Chap.~1.7, (1.7.12)]{Xianda-Matrix-2005}.
        For more details,  refer to \textit{Appendix} \ref{prf:distortion}.
\end{proof}

It is interesting to investigate a system in which there exists an exceptional node, i.e., the channel SNR or/and the observation SNR of the node is very small or very large.
    Without loss of generality, we assume that node $K$ is the exceptional node.
We denote $\mathcal{K}_{-K}=\mathcal{K}-\{K\}$ and the corresponding distortion as $D(\mathcal{K}_{-K})=\sigma^2_\theta \big(a-\frac{c^2}{b}\big)^{-1}$, in which $a=\sum_{k=1}^{K-1} \frac{1}{\lambda_k}$, $b=1+  \sum_{k=1}^{K-1} \frac{u_k^2}{ \lambda_k}$, and $c=\sum_{k=1}^{K-1}\frac{u_k}{\lambda_k}$.
    For such a system, we have the following observations.

\begin{enumerate}
  \item If $\gamma_{\text{ch}K}$ goes either to zero or infinity, or $\gamma_{\text{ob}K}$ goes to infinity, $D(\mathcal{K})$ would be slightly smaller than $D(\mathcal{K}_{-K})$.
      In particular, we have
      \begin{align}
            D(\mathcal{K}) &= \sigma^2_\theta \Big(a-\frac{c^2}{b} + \frac{(b-c)^2}{b^2} \Big)^{-1} \hspace{-5mm} &\text{if}~\gamma_{\text{ch}K}\rightarrow0; \\
            D(\mathcal{K}) &= \sigma^2_\theta \Big(a-\frac{c^2}{b} + \gamma_{\text{ob}K} \Big)^{-1} &\text{if}~\gamma_{\text{ch}K}\rightarrow\infty.
      \end{align}
  \item If $\gamma_{\text{ob}K}$ goes to zero, it can readily be shown that $D(\mathcal{K}) = D(\mathcal{K}_{-K})$, regardless of channel SNR $\gamma_{\text{ch}K}$.
  \item If both $\gamma_{\text{ob}K}$ and $\gamma_{\text{ch}K}$ are sufficiently large (e.g., larger than some  thresholds $\gamma_{\text{ob,th}}$ and $\gamma_{\text{ch,th}}$), we say that node $K$ is a \textit{capable node} and have $D(\mathcal{K}) \rightarrow 0$.
\end{enumerate}

It is seen that the contribution of a node is highly determined by its observation SNR.

\section{To Code or Not to Code} \label{sec:4_homo}
In this section, we consider a \textit{homogeneous} distributed sensing system in which the nodes and the channels  are the same all over the network.
    For this system, we investigate whether the coded scheme performs better than the uncoded scheme or not.

In particular, we assume that all the nodes have the same observation noise power and the same transmit power, and all the links have the same channel gain and the same noise power.
    Furthermore, we assume that the observation noises of the nodes are i.i.d. random variables, so are the channel noises of the links.
In this section, therefore, we shall omit the node-indexes of related variables and denote them as $\sigma^2_{\text{ob}}$, $P$, and $\sigma^2_{\text{ch}}$, respectively.
    Likewise, the corresponding observation SNR and channel SNR are denoted as $\gamma_{\text{ob}}$ and  $\gamma_{\text{ch}}$, respectively.

\subsection{Distortion of Coded Homogeneous Sensing Systems}\label{coded_K}
First, we present the minimum achievable distortion in coded homogenous  sensing systems  by the following theorem.
\begin{theorem}\label{th:homo_digital}
    For a coded homogeneous sensing system with $K$ nodes, the minimum achievable estimation distortion is %given by
    \begin{align} \label{rt:homo_digital}
        D_{\text{Coded}} = \frac{\sigma^2_\theta}{K}
             \left( \frac{\gamma_{\text{ch}}}  {(1+\gamma_{\text{ch}}) \gamma_{\text{ob}} }
                        + \frac{K+\gamma_{\text{ch}}}  {(1+\gamma_{\text{ch}})^2}
             \right).
    \end{align}
\end{theorem}

\begin{proof}
    See \textit{Appendix} \ref{prf:homo_digital}.
\end{proof}

%From \textit{Theorem} \ref{th:homo_digital}, it much clearer that the system distortion is determined by the observation SNR and the channel SNR, and is monotonically decreasing with them.

From \textit{Theorem} \ref{th:homo_digital}, the scaling law of $D_{\text{Coded}}$ with respect to the number of nodes can be readily obtained, as shown in the following corollary.
\begin{corollary}\label{prop:homo_digital_K_inf}
    $D_{\text{Coded}}$ is linearly decreasing with $K$ and converges to the following constant as $K$ goes to infinity,
    \begin{align} \label{rt:homo_digital_K_inf}
        D_{\text{Coded}}(\infty) = \frac{\sigma^2_\theta}{(1+\gamma_{\text{ch}})^2}.
    \end{align}
\end{corollary}

\begin{proof}
    Since both the observation SNR and channel SNR are finite, \textit{Corollary} \ref{prop:homo_digital_K_inf} follows \textit{Theorem} \ref{th:homo_digital} immediately.
\end{proof}

From \textit{Corollary} \ref{prop:homo_digital_K_inf}, it is seen that the estimation distortion is dominated by the channel SNR and the effect of observation noises vanishes gradually with the increase in $K$.
    That is, we can combat observation noises by using more nodes.
It is also seen that $D_{\text{Coded}}$ does not decrease to zero with the increase in $K$, which is due to the correlation among quantization noises.

When the SNR(s) $\gamma_{\text{ob}}$ or/and $\gamma_{\text{ch}}$ go to zero or/and  infinity,  we also have the following corollary.
\begin{corollary} \label{cor:ob_limits}
    In a coded homogeneous sensing system with $K$ nodes,  as the observation SNRs and channel SNRs go to infinity or zero, the achievable estimation distortion is summarized in the following table.

\begin{savenotes}  %footenote can be added
\begin{table}[htbp]
\centering
    \caption{Distortion $D_{\text{Coded}}$ in limiting cases.}\label{tab:distortion_limiting_hetero}
\begin{tabular}{|c|c|c|c|}
\bottomrule
 \diagbox{$\gamma_{\text{ob}}$}{$\gamma_{\text{ch}}$} &  $\infty$ & finite & $0$\\
\hline
 $\infty$ &  $0$ & $\frac{K+\gamma_{\text{ch}}}{K(1+\gamma_{\text{ch}})^2}\sigma^2_\theta$  & $\sigma^2_\theta$\\
\hline
finite  & $\frac{1}{K\gamma_{\text{ob}}}\sigma^2_\theta$ & \eqref{rt:homo_digital} & $\sigma^2_\theta$\\
\hline
$0$    & $\infty$ & $\infty$ & $\sigma^2_\theta$\\
\toprule
\end{tabular}
\end{table}
\end{savenotes}
\end{corollary}

\begin{proof}
    By considering the limitations of \eqref{rt:homo_digital}, the corollary can be proved readily.
\end{proof}

%%%%%%%%%%%%%%%%%%%%%%%%%%%%%%%
\begin{figure*}[htp]   %cross column: add *

\hspace{-6 mm}
    \begin{tabular}{cc}
    \subfigure[$K=3$]
    {
    \begin{minipage}[t]{0.5\textwidth}
    \centering
    {\includegraphics[width = 3.7in] {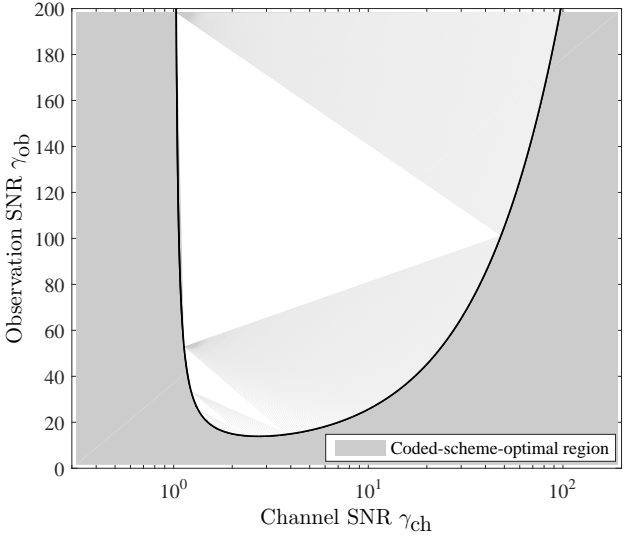} \label{fig:coded_opt_region}}
    \end{minipage}
    }

    \subfigure[$\gamma_{\text{ob}}=0.8$ or $\gamma_{\text{ch}}=0.8$]
    {
    \begin{minipage}[t]{0.5\textwidth}
    \centering
    {\includegraphics[width = 3.7in] {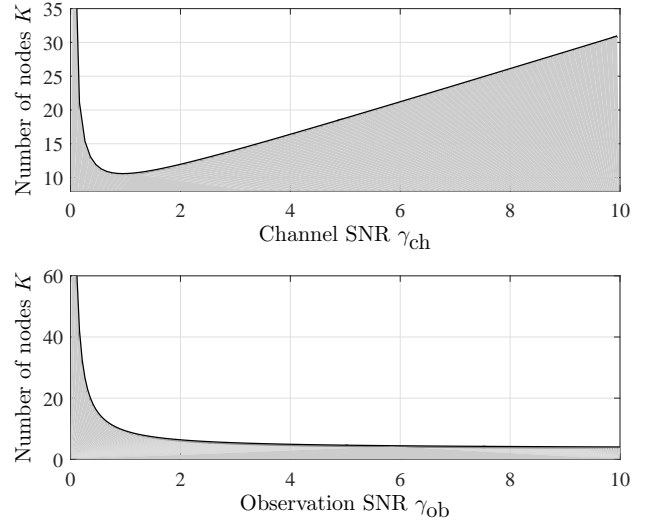} \label{fig:coded_opt_region_vs_K}}
    \end{minipage}
    }
    \end{tabular}
\caption{Coded-scheme-optimal region (the shaded areas). } \label{fig:coded_opt_r}
\end{figure*}
%%%%%%%%%%%%%%%%%%%%%%%%%%%%%%

\subsection{Distortion of Uncoded Homogeneous Sensing Systems}
In this subsection, we consider an uncoded homogeneous sensing system in which the noisy observations are directly amplified and forwarded to the fusion center.
    To be specific, the noisy observation $\widetilde{\theta}_k$ is amplified with power gain
    \begin{equation}\label{df:amplif_gain}
        \alpha = \frac{P}{\sigma^2_\theta+\sigma^2_{\text{ob}}}.
    \end{equation}

The amplified signals are then transmitted to the fusion center through orthogonal Gaussian channels with unit power gain and noise power $\sigma^2_{\text{ch}}$.
    Thus, the observation SNR $\gamma_{\text{ob}}=\sigma^2_\theta/\sigma^2_{\text{ob}}$ and the channel SNR $\gamma_{\text{ch}}=P/\sigma^2_{\text{ch}}$ are the same with those used in Subsection \ref{coded_K}.
Based on the results in \cite{Cui-Estimation-2007}, the minimum achievable distortion in the uncoded homogeneous sensing system can be given by the following proposition.
\begin{proposition}\label{prop:homo_af}
    For an uncoded homogeneous sensing system, the minimum achievable distortion is given by
    \begin{align} \label{rt:homo_af}
        D_{\text{Uncoded}} = \frac{\sigma^2_\theta}{K} \left( \frac{1}{\gamma_{\text{ob}}}
                    + \frac{1}{\gamma_{\text{ch}}}
                    + \frac{1}{\gamma_{\text{ob}}\gamma_{\text{ch}}}
                    \right).
    \end{align}
\end{proposition}

\begin{proof}
    As shown in \cite{Cui-Estimation-2007}, the distortion of an uncoded sensing system is given by
    \begin{equation} \label{eq:distortion_af_cui}
        D_{\text{Uncoded}} = \sigma^2_\theta \left( \sum_{k=1}^K \frac{\alpha'_ks_k}
                {\gamma^{-1}_{\text{ob}k} \alpha'_k s_k +1 } \right)^{-1},
    \end{equation}
    in which $\alpha'_k=P_k/ (1+\gamma^{-1}_{\text{ob}k})$ and $s_k=1/\sigma^2_{\text{ch}k}$.
        With some mathematical manipulations, \textit{Proposition} \ref{prop:homo_af} can be proved.
\end{proof}

As the number $K$ of nodes goes to infinity, it is clear that the distortion of uncoded homogeneous sensing systems would be zero. That is,
    \begin{align} \label{rt:homo_af_K_inf}
        D_{\text{Uncoded}}(\infty) = 0.
    \end{align}

As the observation SNRs or/and the channel SNRs of the nodes go to infinity or zero, the achievable estimation distortion of the uncoded homogeneous system  is  summarized in the following table.

\begin{savenotes}  %footenote can be added
\begin{table}[htbp]
\centering
    \caption{Distortion $D_{\text{Uncoded}}$ in limiting cases.}\label{tab:distortion_limiting_hetero}
\begin{tabular}{|c|c|c|c|}
\bottomrule
 \diagbox{$\gamma_{\text{ob}}$}{$\gamma_{\text{ch}}$} &  $\infty$ & finite & $0$\\
\hline
 $\infty$ &  $0$ & $\frac{1}{K\gamma_{\text{ch}}}\sigma^2_\theta$  & $\sigma^2_\theta$\\
\hline
finite  & $\frac{1}{K\gamma_{\text{ob}}}\sigma^2_\theta$ & \eqref{rt:homo_af} & $\sigma^2_\theta$\\
\hline
$0$    & $\infty$ & $\infty$ & $\infty$\\
\toprule
\end{tabular}
\end{table}
\end{savenotes}

\subsection{Coded Systems Versus Uncoded Systems}\label{subsec:4c}
Based on \textit{Theorem} \ref{th:homo_digital} and \textit{Proposition} \ref{prop:homo_af}, we then investigate the superiority of the coded scheme with respect to the uncoded scheme.

\begin{theorem}\label{th:comparison}
    In a homogeneous sensing system with $K$ nodes and individual power constraints, the coded scheme outperforms the uncoded scheme, i.e., $D_{\text{Coded}} < D_{\text{Uncoded}}$, if
    \begin{itemize}
      \item $K=1$ or $K=2$;
      \item $K\geq 3$ and the observation SNR is small, i.e.,
                \begin{equation} \label{rt:compare_condt_2}
                        \gamma_{\text{ob}} <  \frac{(\gamma_{\text{ch}}+1)(2\gamma_{\text{ch}}+1)} {\max( (K-2)\gamma_{\text{ch}}-1, 0^+) }.
                \end{equation}
    \end{itemize}
\end{theorem}

\begin{proof}
    See \textit{Appendix} \ref{prf:comparison}
\end{proof}

\begin{remark}
    For $K\geq 3$, condition \eqref{rt:compare_condt_2} is equivalent to
    \begin{equation} \label{rt:compare_condt_1}
    \left\{ \hspace{-1.5mm}
    \begin{array}{cc}
        \gamma_{\text{ch}}>0,  &\text{if}~\gamma_{\text{ob}}< \gamma^*_{\text{ob}}, \\
        \gamma_{\text{ch}}>\gamma_{\text{ch}2}~\text{or}~\gamma_{\text{ch}1}>\gamma_{\text{ch}}>0, &\text{if}~\gamma_{\text{ob}}> \gamma^*_{\text{ob}},
    \end{array}
    \right.
    \end{equation}
    in which
    \begin{align} \label{rt:gamma_star}
        \gamma^*_{\text{ob}}&=\frac{3K-2+2\sqrt{ 2(K^2-K) }}  {(K-2)^2}, \\
        \gamma_{\text{ch}1}& = \frac{  (K-2)\gamma_{\text{ob}}-3 - \sqrt{(K-2)^2\gamma^2_\text{ob} - (6K-4)\gamma_\text{ob}+1}  } {4}, \\
        \gamma_{\text{ch}2}& = \frac{  (K-2)\gamma_{\text{ob}}-3 + \sqrt{(K-2)^2\gamma^2_\text{ob} - (6K-4)\gamma_\text{ob}+1}  } {4},
    \end{align}
   as shown in Fig. \ref{fig:coded_opt_region}.
\end{remark}

\begin{remark}\label{rmk:d_k_small}
    For a given SNR pair $(\gamma_{\text{ob}}, \gamma_{\text{ch}})$, the coded scheme outperforms the uncoded scheme if the $K$ satisfies
    \begin{equation}\label{rt:rmk_d_k_small}
        K\leq 2+ \frac{1}{\gamma_{\text{ch}}} +\frac{(\gamma_{\text{ch}}+1)(2\gamma_{\text{ch}}+1)}{\gamma_{\text{ob}}\gamma_{\text{ch}}},
    \end{equation}
\end{remark}
as shown in Fig. \ref{fig:coded_opt_region_vs_K}.

From \textit{Theorem} \ref{th:comparison} and Fig. \ref{fig:coded_opt_r}, it is clear that the coded scheme is more capable of combating noises.
    In particular, we have the following observations.
\begin{enumerate}
  \item [1)] \textit{When} $K$ \textit{is small }(e.g., $K=1,2$), the coded scheme overwhelmingly outperforms the uncoded scheme.
        As $K$ gets larger, however, the coded scheme loses superiority gradually.
             First,  while channel noises contribute to estimation distortion almost directly (cf. \eqref{rt:homo_af}) in uncoded systems, channel noises contribute not too much (cf. \eqref{rt:sigma_qu}, \eqref{rt:homo_digital}) in coded systems, which is beneficial for systems with a few nodes.
        When $K=1$, for example, we have
         \begin{equation} \label{rmk:d_k_1}
                D_{\text{Coded}}= \sigma^2_{\text{ob}} +\frac{\sigma^2_\theta-\sigma^2_{\text{ob}}}{1+\gamma_{\text{ch}}},
        \end{equation}
        which converges to $\sigma^2_\theta$ as $\gamma_{\text{ch}}$ goes to zero.
            For uncoded systems, however, $D_{\text{Unoded}}$ (cf. \eqref{rt:homo_af}) goes to infinity as $\gamma_{\text{ch}}$ goes to zero.
       Second, the margining gain of using one more node in coded systems becomes smaller and smaller as $K$ gets larger.
            The reason is that the correlated quantization noises would strengthen each other at the fusion center.
       In uncoded systems, however, sensing with one more node only brings some additional independent noise.
            The distortion of uncoded systems, therefore, is monotonically decreasing with $K$ and would be smaller than that of coded systems when $K$ is large.

\item [2)] \textit{In the low observation SNR regime} (cf. \eqref{rt:compare_condt_2}), the coded scheme performs better than the uncoded scheme.
        We note that the minimum value of the right hand side of \eqref{rt:compare_condt_2} is $\gamma^*_{\text{ob}}$ (cf. \eqref{rt:gamma_star})  and is achieved with $\gamma_{\text{ch}}=(\sqrt{(K^2-K)/2}+1)/(K-2)$.
     If $\gamma_{\text{ob}}< \gamma^*_{\text{ob}}$ is satisfied, therefore, the coded scheme would outperform the uncoded scheme regardless of node number $K$ and channel SNR $\gamma_{\text{ch}}$.
            For example, in an uncoded system with $K=1$ node (\eqref{rt:compare_condt_2} also applies), the received signal at the fusion center is $y=\sqrt{\alpha}\theta +\sqrt{\alpha} n_{\text{ob}} +n_{\text{ch}} $ and the estimation is $\hat{\theta} = y/\sqrt{\alpha} = \theta + n_{\text{ob}} +n_{\text{ch}}/\sqrt{\alpha}$, in which $\alpha$ is the amplifying gain given in \eqref{df:amplif_gain}.
        Since $\sigma^2_{\text{ob}}$ is very large in the low observation SNR regime, it is clear that $\alpha$ would be very small.
    The distortion $D_{\text{Uncoded}}= \mathbb{E}[(n_{\text{ob}} +n_{\text{ch}}/\sqrt{\alpha})^2]$, therefore, would be very large.
        On the other hand, the estimation noise of a single-node coded system is ${n}={n}_{\text{qu}} - {n}_{\text{ob}}$ and the distortion $D_{\text{Coded}}$ (cf. \eqref{rmk:d_k_1}) can even be smaller than $\sigma^2_{\text{ob}}$, as well as than the distortion of an uncoded system.

  \item [3)] \textit{In the low channel SNR regime} ($\gamma_{\text{ch}}<1/(K-2)$, cf. \eqref{rt:compare_condt_2}), the coded scheme outperforms the uncoded scheme regardless of the observation SNR.
        Note that in coded systems, the channel noise contributes to estimation distortion through quantization noise $\sigma^2_{\text{qu}}=(\sigma^2_\theta+\sigma^2_\text{ob}) /(1+\gamma_{\text{ch}})$, which is upper bounded by $\sigma^2_\theta+\sigma^2_\text{ob}$ and would be smaller than $\sigma^2_{\text{ch}}$ in the low channel SNR regime.
    In an uncoded system, however,  the channel noises contribute to the estimation distortion more directly, which goes to infinity as the channel SNR goes to zero.

  \item [4)]  When the channel SNR is increased, the performance difference between coded systems and uncoded systems becomes smaller and smaller, and goes to zero gradually (cf. \eqref{apx:delta_D}).
  $\hfill{} \blacksquare$
\end{enumerate}

Therefore, the coded scheme is suggested if the observation SNR is low and/or the channel SNR is low;
    in the high channel SNR regime, the uncoded scheme is suggested, since it has almost the same performance as the coded scheme but is much easier to implement.

\subsection{Sensing with Total Power Constraints}

Next, we consider a scenario in which the total transmit power of the $K$ nodes is constrained, i.e., $KP=P_{\text{total}}$ is fixed.
    In particular, we denote the total channel SNR as
    \begin{equation}\label{df:r_total}
        \gamma_{\text{total}} = \frac{P_{\text{total}}}{\sigma^2_{\text{ch}}}.
    \end{equation}
    %It is clear that $\gamma_{\text{total}}=K\gamma_{\text{ch}}$.

\begin{proposition}\label{prop:homo_digital_total}
    For a coded homogeneous sensing system with $K$ nodes and the total power constraint, the minimum achievable distortion is given by
    \begin{align} \label{rt:homo_digital_total}
        D_{\text{Coded}}^{\text{total}}= \sigma^2_\theta\left(
                    \frac{\gamma_{\text{total}}} {K(K+\gamma_{\text{total}})\gamma_{\text{ob}}}
                    + \frac{K^2+ \gamma_{\text{total}}} {(K+\gamma_{\text{total}})^2}
                    \right).
    \end{align}
\end{proposition}

\begin{proof}
    \textit{Proposition} \ref{prop:homo_digital_total} immediately follows \textit{Theorem} \ref{th:homo_digital} and equation \eqref{df:r_total}.
\end{proof}

It can be verified that $D_{\text{Coded}}^{\text{total}}$ is non-convex in node number $K$ and converges to
\begin{equation}\label{rt:homo_digit_total_inf}
    D_{\text{Coded}}^{\text{total}} (\infty)=\sigma^2_\theta
\end{equation}
 as $K$ goes to infinity.
    In particular, the minimum distortion is achieved when $K$ is neither too small nor too large.

For uncoded homogeneous sensing systems with the total transmit power constraint, the corresponding distortion is
    \begin{align} \label{rt:homo_af_total}
        D_{\text{Uncoded}}^{\text{total}}= \sigma^2_\theta\left( \frac{1}{K\gamma_{\text{ob}}} +
                    \frac{1} {\gamma_{\text{total}}}
                    + \frac{1} {\gamma_{\text{ob}}\gamma_{\text{total}}}
                    \right),
    \end{align}
    which converges to
    \begin{equation}\label{rt:homo_af_total_inf}
        D_{\text{Uncoded}}^{\text{total}} (\infty)=\sigma^2_\theta\left(\frac{1} {\gamma_{\text{total}}}
                    + \frac{1} {\gamma_{\text{ob}}\gamma_{\text{total}}}\right)
    \end{equation}
     as $K$ goes to infinity.

Similar to \textit{Theorem} \ref{th:comparison}, we have the following proposition on the distortion of homogeneous systems with the total power constraint and coded/uncoded schemes.
\begin{proposition}\label{prop:comparison_total_p}
    In a homogeneous sensing system with $K$ nodes and the total power constraint, the coded scheme outperforms the uncoded scheme, i.e., $D_{\text{Coded}}^{\text{total}} < D_{\text{Uncoded}}^{\text{total}}$, if
    \begin{itemize}
      \item $K=1$ or $K=2$;
      \item $K\geq 3$ and the observation SNR satisfies
    \end{itemize}
    \begin{equation} \label{rt:compare_condt_total_2}
            \gamma_{\text{ob}} < \frac{(\gamma_{\text{total}}+K)(2\gamma_{\text{total}}+K)}
            {\max( (K^2-2K)\gamma_{\text{total}}-K^2, 0^+)}.
    \end{equation}

\end{proposition}

\begin{proof}
    The proposition can be readily proved either by using \eqref{df:r_total} and \textit{Theorem} \ref{th:comparison} or by combining \eqref{rt:homo_digital_total} and \eqref{rt:homo_af_total}.
\end{proof}

Thus, the comparison between total power constrained coded systems and uncoded systems follows  similar rules with that under the individual power constraint.

\section{ Heterogeneous Sensing and Hybrid Coding}\label{sec:5_hybrid}
In this section,  we investigate the condition for the coded scheme to be optimal in heterogeneous sensing systems with  individual power constraints.
    For a more general case in which the hybrid coding is used, i.e., both the coded scheme and the uncoded scheme are used in the same system, we shall present the corresponding distortion limit and some efficient algorithms for near-optimal policy searching.

\subsection{Coded Optimal Heterogeneous Sensing}
For a heterogeneous sensing system defined by the number $K$ of nodes, the channel SNRs $\{ \gamma_{\text{ch}k}, k\in\mathcal{K} \}$ and the observation SNRs $\{ \gamma_{\text{ob}k}, k\in\mathcal{K} \}$, the coded scheme outperforms the uncoded scheme if the following condition is satisfied.

\begin{theorem}\label{th:condition_hetero}
    In a heterogeneous sensing system with $K$ nodes and the individual power constraint, the coded scheme outperforms the uncoded scheme if
    \begin{equation} \label{rt:condition_hetero}
        \sum_{k=1}^{K} \frac{(1+2\gamma_{\text{ch}k}) \gamma_{\text{ob}k}}
        {(1+\gamma_{\text{ch}k}+\gamma_{\text{ob}k})\gamma_{\text{ch}k}}
        >\left( \sum_{k=1}^{K} \frac{ \gamma_{\text{ob}k}}
        {1+\gamma_{\text{ch}k}+\gamma_{\text{ob}k}}
        \right)^2.
    \end{equation}
\end{theorem}

\begin{proof}
    For heterogeneous sensing systems, the estimation distortion of the coded scheme is given by \textit{Theorem} \ref{th:distortion} while the estimation distortion of the uncoded scheme can be obtained using \eqref{eq:distortion_af_cui}.
        By comparing the two distortions, the condition shown in \eqref{rt:condition_hetero} can be obtained readily.
    For more details,  refer to \textit{Appendix} \ref{prf:condition_hetero}.
\end{proof}

It can be verified that the condition given by \eqref{rt:condition_hetero} is equivalent to
    \begin{equation} \label{rt:condition_hetero_equiverlent}
        \sum_{k=1}^{K} \frac{ \gamma_{\text{ob}k}}
        {(1+\gamma_{\text{ch}k}+\gamma_{\text{ob}k})\gamma_{\text{ch}k}} +1
        >\left( \sum_{k=1}^{K} \frac{ \gamma_{\text{ob}k}}
        {1+\gamma_{\text{ch}k}+\gamma_{\text{ob}k}} -1
        \right)^2.
    \end{equation}

\begin{remark}
From \eqref{rt:condition_hetero_equiverlent}, we observe that the uncoded scheme would perform better if:
\begin{enumerate}
  \item $K$ is large. Note that the left-hand-side increases linearly with $K$ while the right-hand-side increases quadratically with $K$ (approximately);
  \item channel SNRs are large. If $\gamma_{\text{ch}k}$ is increased, it is clear that the left-hand-side would be decreased dramatically.
\end{enumerate}
\end{remark}

\subsection{Hybrid Coding}
As shown in \textit{Theorem} \ref{th:comparison} and \textit{Theorem} \ref{th:condition_hetero}, either the coded scheme or the uncoded scheme can be optimal under  a certain conditions.
    This motivates us to consider a hybrid coding for heterogeneous sensing systems, in which  $K_1$ nodes use the coded scheme,  $K_0$ node uses the uncoded scheme, and $K_1+K_0=K$.
In particular, we present the corresponding minimum achievable distortion in the following theorem.

\begin{theorem}\label{th:distortion_hybrid}
    The minimum achievable distortion of a heterogeneous sensing system with $K_1$ coded nodes and $K_0$ uncoded nodes is given by
    \begin{align}
    \hspace{-3mm}\nonumber
        D_{\text{Hybrid}} =& \sigma^2_\theta\Bigg(
        \sum_{k=1}^{K_1} \frac{1}{\lambda_k} - \frac{1}{1+  \sum_{k=1}^{K_1} \frac{u_k^2}{ \lambda_k} }
        \Bigg(\sum_{k=1}^{K_1}\frac{u_k}{\lambda_k} \Bigg)^2  \\
        \label{rt:distortion_hetero}
        & \left. \qquad\qquad\quad + \sum_{k=1}^{K_0}\frac{1}{
                    \frac{1}{\gamma_{\text{ob}k}}
                    + \frac{1}{\gamma_{\text{ch}k}}
                    + \frac{1}{\gamma_{\text{ob}k}\gamma_{\text{ch}k}}
                    }
                \right)^{-1},
    \end{align}
    in which $u_k$ and $\lambda_k$ are given by \eqref{rt:u_k} and \eqref{rt:lamda_k}, respectively.
\end{theorem}

\begin{proof}
    See \textit{Appendix} \ref{prf:distortion_hybrid}.
\end{proof}

With \textit{Theorem} \ref{th:distortion_hybrid} we can further minimize the distortion of the system by optimally selecting the coding scheme for each node, which is an assignment problem with binary choices.
    From \eqref{rt:distortion_hetero}, however, it is seen that by using one more coded node, the marginal gain of $\sigma^2_\theta D_{\text{Hybrid}}^{-1}$ is closely related to the number and the SNRs of the other coded nodes, and thus is not a constant.
Without a deterministic matrix of marginal gains, this assignment problem is much more difficult to solve than conventional assignment problems.
    In this paper, therefore, we shall solve the problem by the following greedy algorithms.

\subsection{Optimal and Greedy Algorithms}\label{subsec:greedy_algo}
In this subsection, we consider the following four searching algorithms for the optimal coding schemes of nodes: the global searching algorithm, the pure greedy algorithm, the group greedy algorithm, and the sorted greedy algorithm.
    To perform these algorithms at the fusion center, we assume that the fusion center can access the channel gains of each link (e.g., by channel estimation techniques in \cite{chanelesti-TSP-2010,chanelesti-TSP-2018}).
When the coding policy has been determined, the fusion center will notify each node of its coding scheme through an one-bit feedback, which indicates whether the coded scheme or the uncoded scheme should be used.

We denote a coding policy as $\varrho = \{\rho_1, \rho_2,\cdots,\rho_K\}$, in which $\rho_k=0$ if the uncoded scheme is used and $\rho_k=1$ if the coded scheme is used.
    We denote the set of all feasible policies as searching space $\mathcal{P}$.
Since each $\rho_k$ has two choices, it is clear that $|\mathcal{P}|=2^K$, which increases exponentially with the number of nodes.

\subsubsection{Global Searching}
To perform this algorithm, the distortion of each feasible policy $\varrho\in \mathcal{P}$ will be evaluated using \eqref{rt:distortion_hetero}.
    Thus, the optimal policy can be found surely.
However, the computational complexity of global searching is high, especially when $K$ is large.

\subsubsection{Pure Greedy}
The pure greedy algorithm finds its solution by iteratively updating an \textit{active set} $\mathcal{A}_{\text{p}}$  and the corresponding policy $\varrho_{\text{p}}$.
     To be specific, the algorithm expands $\mathcal{A}_{\text{p}}$ with one more node in each iteration and stops when all the nodes have been included, i.e., $\mathcal{A}_{\text{p}}=\mathcal{K}$.
In particular, the new node $k^*$ added to $\mathcal{A}_{\text{p}}$ is chosen as the one reducing the distortion most.
That is,
\begin{equation} \label{cond:local_opt_pure}
    (k^*, \rho_{k^*}) = \arg \min\limits_{k\in \mathcal{K}-\mathcal{A}_{\text{p}}, \rho_k \in\{0,1\}} D_{\text{Hybrid}} (\varrho_{\text{p}}\cup  \rho_{k}),
\end{equation}
    in which $D_{\text{Hybrid}} (\varrho_{\text{p}}\cup  \rho_{k})$ is the distortion (cf. \eqref{rt:distortion_hetero}) of the sub-system with node set $\mathcal{A}_{\text{p}}\cup k$ and policy $\varrho_{\text{p}}\cup  \rho_{k}$.

The pure greedy algorithm is shown in \textit{Algorithm} \ref{alg:p_greedy}, in which the output $\mathcal{A}_{\text{p}}$ specifies the node-order of policy $\varrho_{\text{p}}$.
    Since \eqref{rt:distortion_hetero} can be calculated with $O(K)$ operations, the computational complexity of \textit{Algorithm} \ref{alg:p_greedy} would be $O(K^3)$.

When the locally optimal node $k^*$ and its coding scheme $\rho_{k^*}$ have been determined to update the estimation distortion, we note that the denominator of \eqref{rt:distortion_hetero} does not increase additively.
    Thus, policy $\varrho_{\text{p}}\cup  \rho_{k^*}$ is most probably not globally optimal for node set $\mathcal{A}_{\text{p}}\cup k^*$.
To be specific, the globally optimal coding policy is not a simple combination of $\varrho_{\text{p}}$ (even if it is optimal for $\mathcal{A}_{\text{p}}$) and $\rho_{k^*}$, but a brand new policy reconsidered for each node of subset $\mathcal{A}_{\text{p}}\cup k^*$.
    Nevertheless, the output $D_{\text{Hybrid}} (\varrho_{\text{p}})$ of \textit{Algorithm} \ref{alg:p_greedy} approaches the distortion of global searching quite well, as shown in Figs. \ref{fig:d_ratio_vs_K}  and \ref{fig:d_ratio_vs_L}, Section \ref{sec:6_simulation}.

\begin{algorithm}[!t]
\algsetup{linenosize=\small}
\scriptsize
\caption{Pure greedy algorithm}
\begin{algorithmic}[1]\label{alg:p_greedy}
\REQUIRE ~~\\%Initialization
    \STATE Set $\mathcal{A}_{\text{p}}=\emptyset$ and $\varrho_{\text{p}}=\emptyset$;
\ENSURE ~~\\%Iteration
\FOR {$j=1$ to $K$}
    \STATE Find the optimal node $k^*$ and policy $\rho_{k^*}$ using \eqref{cond:local_opt_pure};
    \STATE $\mathcal{A}_{\text{p}} = \mathcal{A}_{\text{p}} \cup k^*$;  \label{step:active_set_update}
    \STATE $\varrho_{\text{p}} = \varrho_{\text{p}} \cup  \rho_{k^*}$; \label{step:policy_update}
\ENDFOR
 \STATE \textbf{Output:} $\mathcal{A}_{\text{p}}, \varrho_{\text{p}}, D_{\text{Hybrid}} (\varrho_{\text{p}})$.
\end{algorithmic}

\end{algorithm}

\subsubsection{Group Greedy}
Motivated by \cite{Soh-Group-2019}, we shall achieve a good balance between computational complexity and distortion performance by the following group greedy algorithm.
    In each iteration of the algorithm, to be specific, $L$ potential sub-policies $\{\varrho_{\text{g}l}, l=1,\cdot\cdot, L\}$ and the corresponding node sets $\{\mathcal{L}_{\text{g}l}, l=1,\cdot\cdot, L\}$ are searched to preserve the potential policy towards the optimal solution.
Based on each potential sub-policy $\varrho_{\text{g}m}$ obtained in the previous iteration, we then calculate the distortion $D_{\text{Hybrid}} (\varrho_{\text{g}m}\cup  \rho_{k})$ for subset $\mathcal{L}_{\text{g}m}\cup k$ using \eqref{rt:distortion_hetero} for all $k\in \mathcal{K}- \mathcal{L}_{\text{g}m}$ and $\rho_k\in \{0,1\}$.
Afterwards, the best $L$ node-scheme pairs are selected by
\begin{align} \label{cond:local_opt_group}
\nonumber
%\hspace{-2mm}
    \{(k_l, \rho_{k_l}),l\in&\{1,\cdot\cdot,L\}\}\\
     = \arg &\mathop{\small\texttt{Lminimum}}_
     {
    \mbox{\tiny
    $\begin{array}{c}
        m\in\{1,\cdots, L\}  \\
        k\in \mathcal{K}-\mathcal{L}_{\text{g}m}, \rho_k \in\{0,1\}
      \end{array}$}
      }
      D_{\text{Hybrid}} (\varrho_{\text{g}m}\cup  \rho_{k}),
\end{align}
in which \texttt{Lminimum} is a function sorting a sequence in ascending order and returns its smallest $L$ elements.
    Finally, the $L$ new potential policies and corresponding node sets are updated by $\varrho_{\text{g}l}= \varrho_{\text{g}m_l}\cup \rho_{k_l}$ and $\mathcal{L}_{\text{g}l}=\mathcal{L}_{\text{g}m_l}\cup k_l$ for each $l\in\{1,\cdots,L\}$, in which $m_l$ is the index of the previous potential policy based on which $(k_l, \rho_{k_l})$ is obtained.
When all the nodes have been included in each of the potential node sets, the policy with smallest distortion is considered as the final output, as shown in \textit{Algorithm} \ref{alg:g_greedy}.

If $L=O(K)$, the computational complexity of the group greedy algorithm would  be $O(K^3)$, which is reasonably small.
    In particular, the group policy degrades to the pure greedy policy when $L=1$.
On the other hand, since there are no more than ${K\choose k}2^k$ potential sub-policies when $k$ nodes has been included in each potential node set, the optimal policy would certainly be found if we set $L=\max_{k} {K\choose k}$.
    In this case, however, the computational complexity is very high.
In fact, the complexity is even higher than that of global searching, since the optimal policy for many subsets of $\mathcal{K}$ are also considered during the iterations.
    Nevertheless, our results show that the group greedy algorithm achieves satisfying performance when $L$ is reasonably small, as is shown in Figs. \ref{fig:d_ratio_vs_L} and \ref{fig:policy_error_ratio_vs_L}, Section \ref{sec:6_simulation}.

\begin{algorithm}[!t]
\algsetup{linenosize=\small}
\scriptsize
\caption{Group greedy algorithm}
\begin{algorithmic}[1]\label{alg:g_greedy}
\REQUIRE ~~\\%Initialization
    \STATE Set $\mathcal{L}_{\text{g}l}=\emptyset$ and $\varrho_{\text{g}l}=\emptyset, \forall l\in\{1,\cdots,L\}$;
\ENSURE ~~\\%Iteration
\FOR {$j=1$ to $K$}
    \STATE Find the best $L$ potential policies using \eqref{cond:local_opt_group}.
    \STATE $\mathcal{L}_{\text{g}l}=\mathcal{L}_{\text{g}m_l}\cup k_l, \forall l\in\{1,\cdots,L\}$;
    \STATE $\varrho_{\text{g}l}= \varrho_{\text{g}m_l}\cup \rho_{k_l}, \forall l\in\{1,\cdots,L\}$;
\ENDFOR
 \STATE \textbf{Output:} $\mathcal{A}_{\text{g}}^* = \mathcal{L}_{\text{g}1}, \varrho_{\text{g}}^*=\varrho_{\text{g}1} ,
                                             D_{\text{Hybrid}} (\varrho_{\text{g}1})$.
\end{algorithmic}

\end{algorithm}

\subsubsection{Sorted Greedy}
Unlike the group greedy algorithm which improves performance at a cost of larger searching space and higher complexity, we propose a sorted greedy algorithm in this subsection to achieve satisfying performance with much lower computational complexity.

The main idea is that many computations could be avoided if we can specify the order of adding nodes to the active set $\mathcal{A}_{\text{s}}$  in advance.
    To this end, we shall take the distortion achieved by each single node as a sorting criterion in the initialization phase.
Since \textit{Theorem} \ref{th:comparison} has shown that the coded scheme outperforms the uncoded scheme for $K=1$,
    we shall calculate the distortion $D_k$ of each individual node with equation \eqref{rt:distortion} and $\rho_k=1$ (i.e., the coded scheme).
After that, the distortion sequence $\{D_k, k\in\mathcal{K}\}$ would be sorted in descending order $\boldsymbol{\pi}=\{\pi_1,\cdot\cdot,\pi_K\}$.
    Since the coded scheme performs better in the low SNR regime, we shall take the node $\pi_1$ (with the largest distortion)  as the first active node of $\mathcal{A}_{\text{s}}$ and set $\rho_{\pi_1}=1$.
In the following operations, we shall add the node with the $k$-th largest distortion (i.e., $\pi_k$) to $\mathcal{A}_{\text{s}}$ in the $k$-th iteration and determine its coding scheme via distortion comparison (cf. step \ref{step:distortion_compare}).
    The outline of the sorted greedy algorithm is shown in \textit{Algorithm} \ref{alg:s_greedy}.
It is seen that the computational complexity of the sorted greedy algorithm is $O(K^2)$, which is lower than that of  pure greedy algorithm.

\begin{algorithm}[!t]
\algsetup{linenosize=\small}
\scriptsize
\caption{Sorted greedy algorithm}
\begin{algorithmic}[1]\label{alg:s_greedy}
\REQUIRE ~~\\%Initialization
    \STATE Calculate distortion $D_k$ using \eqref{rt:homo_af} for all $k\in\mathcal{K}$;
    \STATE Sort $\{D_k, k\in\mathcal{K}\}$ in descending order and return order $\boldsymbol{\pi}$;
    \STATE Set $\mathcal{A}_{\text{s}}=\pi_1$ and $\varrho_{\text{s}}=1$;
\ENSURE ~~\\%Iteration
\FOR {$k=2$ to $K$}
    \STATE $\mathcal{A}_{\text{s}} = \mathcal{A}_{\text{s}} \cup \pi_k$;
    \STATE $\rho_{\pi_k}^* = \min\limits_{\rho_{\pi_k} \in\{0,1\}} D_{\text{Hybrid}} (\varrho_{\text{s}}\cup  \rho_{\pi_k})$; \label{step:distortion_compare}

    \STATE $\varrho_{\text{s}} = \varrho_{\text{s}} \cup  \rho_{\pi_k}^*$; \label{step:policy_update_sorted}
\ENDFOR
 \STATE \textbf{Output:} $\mathcal{A}_{\text{s}}, \varrho_{\text{s}}, D_{\text{Hybrid}} (\varrho_{\text{s}})$.
\end{algorithmic}

\end{algorithm}

Our simulation results show that the  performance of the sorted greedy algorithm closely approaches that of global searching (cf. Figs. \ref{fig:d_ratio_vs_K} and \ref{fig:policy_error_ratio_vs_K}), especially when $K$ is large.

\begin{figure}[!t]
\centering
\includegraphics[width=3.7in]{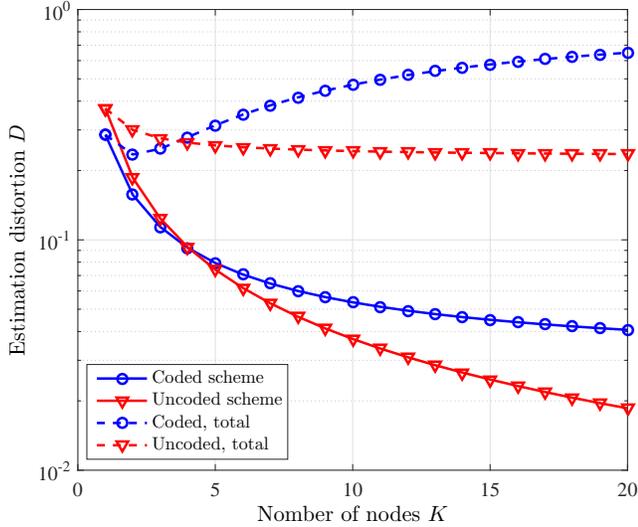}
\caption{Estimation distortion  versus number of nodes in homogeneous systems.} \label{fig:d_vs_K}
\end{figure}

%%%%%%%%%%%%%%%%%%%%%%%%%%%%%%%
\begin{figure}[htp]   %cross column: add *

\hspace{-6 mm}
    \begin{tabular}{cc}
    \subfigure[$K=3$]
    {
    \begin{minipage}[t]{0.5\textwidth}
    \centering
    {\includegraphics[width = 3.7in] {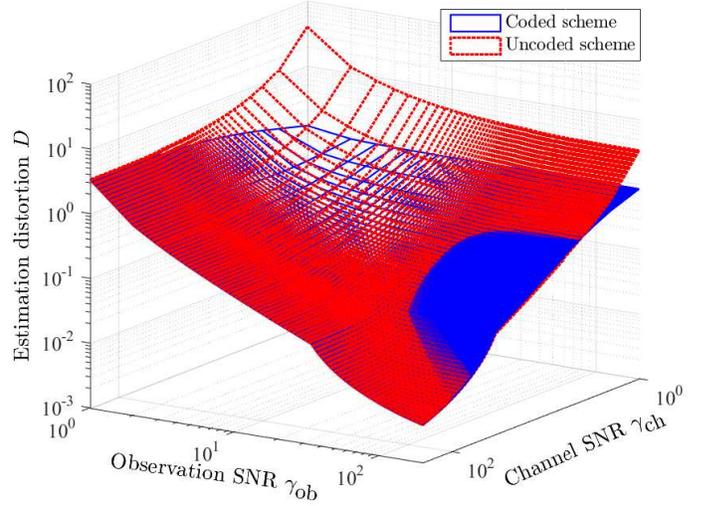} \label{fig:distortion_k3}}
    \end{minipage}
    }\\

    \subfigure[$K=30$]
    {
    \begin{minipage}[t]{0.5\textwidth}
    \centering
    {\includegraphics[width = 3.7in] {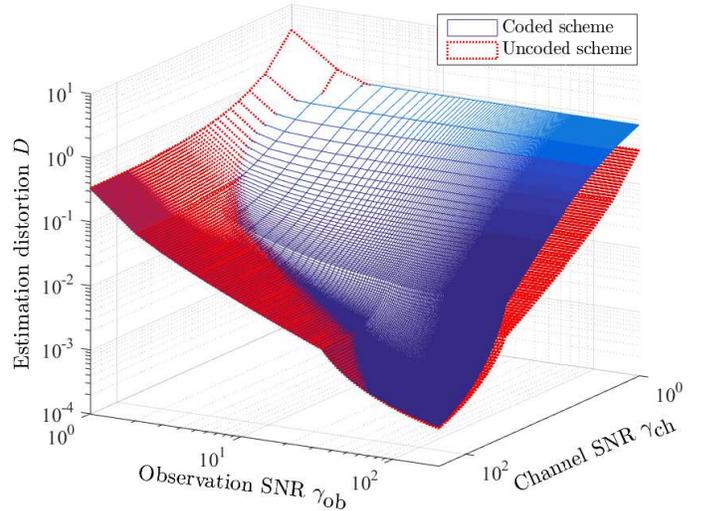} \label{fig:distortion_k30}}
    \end{minipage}
    }
    \end{tabular}
\caption{Estimation distortion versus channel SNR and observation SNR  in homogeneous systems. } \label{fig:distortion_snr}
\end{figure}
%%%%%%%%%%%%%%%%%%%%%%%%%%%%%%

\section{Numerical Results}\label{sec:6_simulation}
In this section, we present the obtained results through numerical and Monte Carlo simulations.
    We set the source signal power to $\sigma^2_\theta=1$ and express SNRs in the non-decibel format.

\subsection{Distortion of Homogeneous Sensing Systems}
In Fig. \ref{fig:d_vs_K}, we present  the scaling law of the estimation distortion of homogeneous sensing systems (cf. Section \ref{sec:4_homo}), in which the coded scheme or the uncoded scheme is used, with the individual or the total power constraint.
    We set the observation SNR to $\gamma_{\text{ob}}=7$ and the channel SNR to $\gamma_{\text{ch}}=5$.
For systems with the individual power constraint, we  observe that the estimation distortion is decreasing with $K$ both when the coded scheme is used (the solid curve with $\circ$, cf. \eqref{rt:homo_digital}) and when the uncoded scheme is used (the solid curve with $\triangledown$, cf. \eqref{rt:homo_af}).
    This is because when $K$ is increased, the system would have more energy and diversity to perform the estimation.
In particular, it is seen that as $K$ goes to infinity, the distortion  goes to zero (cf. \eqref{rt:homo_af_K_inf}) in uncoded systems and converges to some none-zero constant (cf. \eqref{rt:homo_digital_K_inf}) in coded systems.
    The reason is that in coded systems, the marginal gain of using more nodes is significantly constrained by the correlation among quantization noises.
We also observe that the coded scheme outperforms the uncoded scheme when $K$ is small and underperforms the uncoded scheme when $K$ is large, as shown in  \textit{Remark} \ref{rmk:d_k_small}.
    More specifically, the uncoded scheme performs better if $K>4.0857$ (cf. \eqref{rt:rmk_d_k_small}) for the setup considered here.

When the nodes are constrained by their total power, we set $\gamma_{\text{total}}=5$ and $\gamma_{\text{ob}}=7$.
 As shown by the dashed curves, the corresponding estimation distortions are larger than those under the individual power constraint.
    Particularly, the distortion of  a coded system is the smallest when $K$ is small (e.g., $K=2$) and converges to $\sigma^2_\theta$ as $K$ goes to infinity (cf. \eqref{rt:homo_digital_total}, \eqref{rt:homo_digit_total_inf}).
On the other hand, the distortion of an uncoded system monotonically decreases to some constant (cf. \eqref{rt:homo_af_total_inf}) as $K$ goes to infinity.
    In a nutshell, the coded scheme would most probably perform better if $K$ is small; the uncoded scheme is suggested if plenty of nodes are available (see also in \textit{Proposition} \ref{prop:comparison_total_p} and \textit{Remark} \ref{rmk:d_k_small}).
For the setup considered here, the uncoded scheme performs better if $K>3.65483$ (cf. \eqref{rt:compare_condt_total_2}).

Fig. \ref{fig:distortion_snr} depicts how the channel SNR and the observation SNR affect the estimation distortion of homogeneous sensing systems with the individual power constraint.
    It is seen that the distortion becomes smaller when either the channel SNR or the observation SNR is increased.
Moreover, the uncoded scheme  outperforms (has smaller distortion) the coded scheme only if the observation SNR is sufficiently large and the channel SNR is neither too small nor too large, as shown in \textit{Theorem} \ref{th:comparison}.
    It is also observed that the uncoded scheme outperforms the coded scheme in more cases  when $K$ is large (e.g., $K=30$).
As $K$ gets larger, however, the difference between their performances becomes smaller and smaller.

\begin{figure}[!t]
\centering
\includegraphics[width=3.7in]{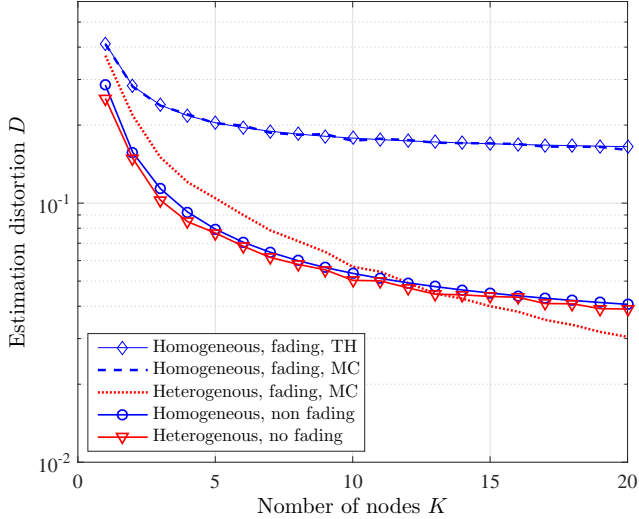}
\caption{Estimation distortion of coded sensing systems with and without fading.} \label{fig:D_fading}
\end{figure}

\subsection{Distortion of Fading Sensing Systems} \label{subcec:fading_sim}
In this subsection, we investigate the estimation distortion of distributed sensing systems with individual power constraints and block Rayleigh fading.
    To be specific, the fading power gain $h$ keeps unchanged during the transmission of each channel codeword and varies randomly among different transmissions.
Also,  the fading power gain of each node is independent from those of other nodes.
    For the given average fading power gain $\nu>0$, the probability density function of $h$ is given by $f_{\text{h}}(x)=\frac1\nu \exp(\frac{-x}{\nu})$.

 Note that the instantaneous channel SNR can be expressed as $\gamma_{\text{ch}k}^{\text{fading}}=h_k\gamma_{\text{ch}k}$, in which $\gamma_{\text{ch}k}=\frac{P_k}{\sigma^2_{\text{ch}k}}$ is constant over time and is different among nodes.
    In each round of the Monte Carlo simulation, we generate a sequence of channel SNR $\{\gamma_{\text{ch}k}, k\in \mathcal{K}\}$ randomly from a folded normal distribution with standard  deviation $\sigma_1$.
Likewise, we  generate the observation SNRs $\{\gamma_{\text{ob}k}, k\in \mathcal{K}\}$ according to a folded normal distribution with standard  deviation $\sigma_2$.
    In particular, the location parameters $\mu$ of the two distributions are chosen such that the arithmetic means of $\gamma_{\text{ch}k}$ and $\gamma_{\text{ob}k}$ are equal to $\gamma_{\text{ch}}$ and $\gamma_{\text{ob}}$ respectively, i.e., the SNR parameters for the non-fading homogeneous sensing system.

In the homogenous case, the instantaneous channel SNR of the system is given by $\gamma_{\text{ch}}^{\text{fading}} = h\gamma_{\text{ch}}$, in which $h$ will change randomly for each period of transmission.
    Thus, the estimation distortion should be the statistical average of the instantaneous distortions given by \eqref{rt:homo_digital} and \eqref{rt:homo_af}.
In particular, the estimation distortion of a coded fading homogeneous system is given by
    \begin{align} \nonumber
        D_{\text{Coded}}^{\text{fading}} &= \frac{\sigma^2_\theta}{K}
             \left( \frac{1}  {\gamma_{\text{ob}} } +
             \frac{(\gamma_{\text{ob}}-1) e^{\frac{1}{\nu \gamma_{\text{ch}}}} }
                        {\nu\gamma_{\text{ob}} \gamma_{\text{ch}} }
             \text{E}_1\Big(\frac{1}{\nu \gamma_{\text{ch}}}\Big)
             \right. \\
        \label{rt:homo_digital_bf}
        &\hspace{1.95cm} \left.  + \frac{(K-1) e^{\frac{1}{\nu \gamma_{\text{ch}}}} } {\nu \gamma_{\text{ch}}}   \text{E}_2\Big(\frac{1}{\nu \gamma_{\text{ch}}}\Big)
             \right),
    \end{align}
   in which $\text{E}_n(x)=\int_1^\infty \frac{1}{t^n} e^{-xt} \text{d}t$ is the exponential integral of order $n$.
        When the uncoded scheme is used, however, the corresponding distortion would be $D_{\text{Uncoded}}^{\text{fading}} = \infty$.
   This is because the probability that $h$ is close to zero is strictly positive and the corresponding estimation distortion (cf. \eqref{rt:homo_af}) approaches infinity.

    We set $\nu=0.9$, $\gamma_{\text{ch}}=5$, $\gamma_{\text{ob}}=7$, and $\sigma_1=\sigma_2=1.5$.
First, the estimation distortion of the fading homogeneous sensing system is shown theocratically (TH)  by the $\diamond$-labeled curve (cf. \eqref{rt:homo_digital_bf}), which coincides with the corresponding Monte Carlo (MC) result (the dash-dotted curve) exactly.
    For the fading heterogeneous sensing system, the Monte Carlo result is presented by the dotted curve and no explicit theoretical result is available.
 It is observed that in the fading heterogeneous sensing system, the distortion is much smaller than that of the fading homogeneous sensing system.
    The main reason is that for the fading heterogeneous sensing system, it is more likely to have some capable (i.e., with a large channel SNR and a large observation SNR) nodes, which can reduce the estimation distortion significantly.
 Second, it is noted from Fig. \ref{fig:distortion_snr} and \textit{Proposition} \ref{prop:homo_digital_total} that although $D_{\text{Coded}}$ is not a convex function of $\gamma_{\text{ch}}$, it does not deviate from a convex function very much.
    The estimation distortion of the fading homogenous sensing system, therefore, would be reduced if the instantaneous channel SNRs are almost the same in each period of transmission, i.e., if the randomness in the fading gain is reduced.
 Thus, the estimation distortion of the non-fading homogeneous sensing system (with channel SNR $\gamma_{\text{ch}}$ and observation SNR $\gamma_{\text{ob}}$),  would be much smaller, as shown by the $\circ$-labeled curve (cf. \eqref{rt:distortion}).
    Third, the estimation distortion of the non-fading heterogeneous sensing system (in which $\{\gamma_{\text{ch}k}\}$ and $\{\gamma_{\text{ob}k}\}$ are used) is shown by the curve labeled with $\triangledown$ (cf. \eqref{rt:homo_digital}).
 As is shown, the estimation distortion is slightly smaller than that of the non-fading homogeneous sensing system, which is due to the randomness in the SNRs of the links in the heterogeneous system.
    It also noted that the non-fading heterogeneous sensing system outperforms the fading heterogeneous sensing system only when $K$ is small.
This is because when $K$ is large, the probability of having a capable node is also larger.
    \footnote{The probability for the fading heterogeneous system to have a capable node can be expressed as $p_{\text{sys}}=K p_{\text{c}}$, in which $p_{\text{c}}$ is the probability for the observation SNR and the channel SNR of a node to be larger than some given thresholds, i.e., to be capable.
 It is clear that $p_{\text{sys}}$ is increasing with $K$.}

\begin{figure}[!t]
\centering
\includegraphics[width=3.7in]{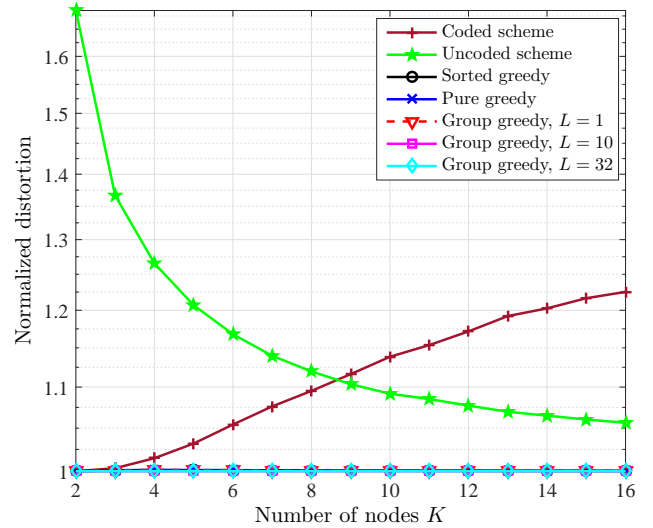}
\caption{Normalized distortion of the coded scheme, the uncoded scheme, and the proposed hybrid schemes.} \label{fig:D_hybrid}
\end{figure}

%%%%%%%%%%%%%%%%%%%%%%%%%%%%%%%
\begin{figure*}[htp]   %cross column: add *

\hspace{-6 mm}
    \begin{tabular}{cc}
    \subfigure[Normalized distortion versus number of nodes]
    {
    \begin{minipage}[t]{0.5\textwidth}
    \centering
    {\includegraphics[width = 3.7in] {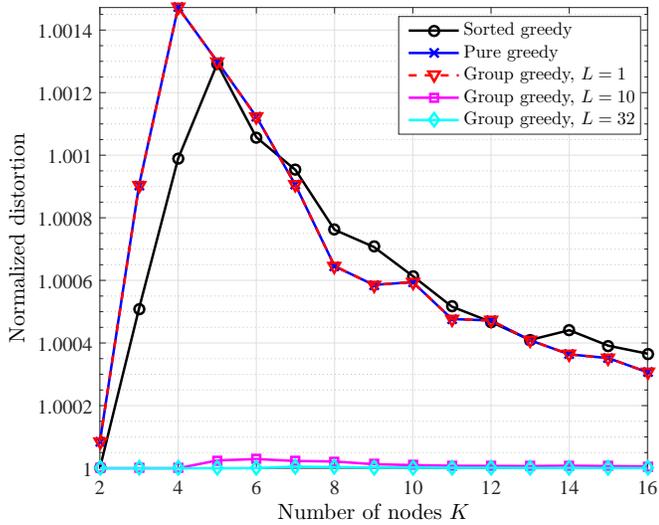} \label{fig:d_ratio_vs_K}}
    \end{minipage}
    }

    \subfigure[Policy error rate versus number of nodes]
    {
    \begin{minipage}[t]{0.5\textwidth}
    \centering
    {\includegraphics[width = 3.7in] {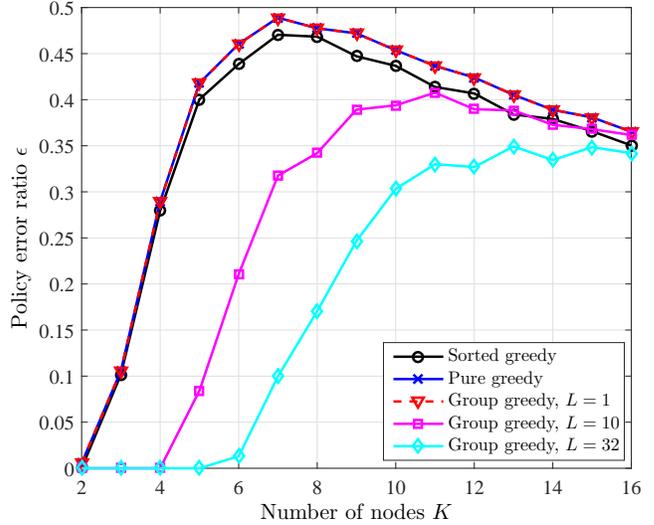} \label{fig:policy_error_ratio_vs_K}}
    \end{minipage}
    }\\

    \subfigure[Normalized distortion versus group size ($K=10$)]
    {
    \begin{minipage}[t]{0.5\textwidth}
    \centering
    {\includegraphics[width = 3.7in] {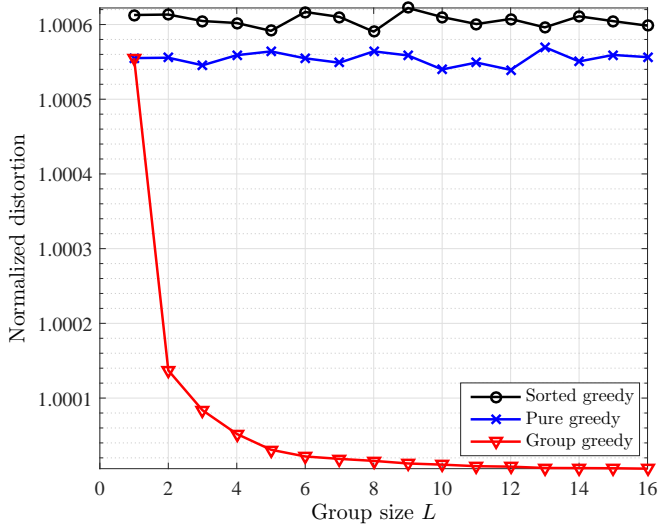} \label{fig:d_ratio_vs_L}}
    \end{minipage}
    }

    \subfigure[Policy error rate versus group size ($K=10$)]
    {
    \begin{minipage}[t]{0.5\textwidth}
    \centering
    {\includegraphics[width = 3.7in] {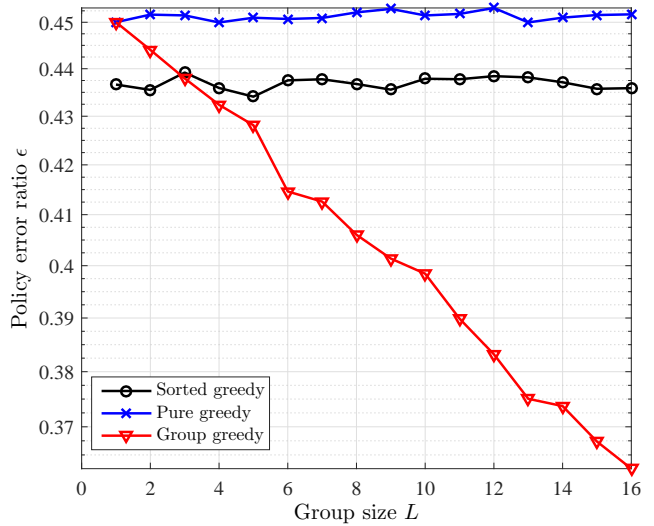} \label{fig:policy_error_ratio_vs_L}}
    \end{minipage}
    }

    \end{tabular}
\caption{Normalized distortion and policy error rate of greedy algorithms. } \label{fig:distortion}
\end{figure*}
%%%%%%%%%%%%%%%%%%%%%%%%%%%%%%

\subsection{Performance of Greedy Algorithms}
In this part, we evaluate the performance of the proposed greedy algorithms (cf. Subsection \ref{subsec:greedy_algo}).
    First, we propose the following two metrics  to evaluate the performance of the algorithms under test.

\begin{definition}
     \textit{Normalized distortion} $ \widetilde{D}$ is the ratio between the average distortion $\widebar{D}$ of a greedy algorithm and the average distortion $\widebar{D}_{\text{opt}}$ of global searching. That is,
    \begin{equation} \label{df:normal_D}
        \widetilde{D} = \frac{\widebar{D}}{\widebar{D}_{\text{opt}}}.
    \end{equation}
\end{definition}

\begin{definition}
     \textit{Policy error rate} $\epsilon$ is defined as the probability that an element of the policy $\varrho$ obtained by a greedy algorithm is different from that of the optimal policy $\varrho_{\text{opt}}$ obtained by global searching.
        Thus, $\epsilon$ can be empirically calculated by
    \begin{equation} \label{df:error_rate}
        \epsilon = \frac{N_{\text{error}}}{N_{\text{total}}},
    \end{equation}
    in which $N_{\text{total}}=N_{\text{sim}}K$ and $N_{\text{error}}$ is the total number of errors in policy $\varrho_n$ compared with policy $\varrho_{\text{opt}n}$ in $N_{\text{sim}}$ rounds of simulations.
Specifically, we have  $N_{\text{error}}=\sum_{n=1}^{N_\text{sim}} \sum_{k=1}^K\rho_{nk}\oplus \rho_{\text{opt}nk}$, in which $\oplus$ is the bitwise modulo-two sum operator.
\end{definition}

Figs. \ref{fig:D_hybrid}, \ref{fig:d_ratio_vs_K}, and \ref{fig:policy_error_ratio_vs_K} present how the normalized distortion and the policy error rate change with the number of nodes.
    We run each algorithm for $N_{\text{sim}}=10000$ rounds and calculate $\widetilde{D}$ and $\epsilon$, respectively, according to \eqref{df:error_rate} and \eqref{df:normal_D}.

In Fig. \ref{fig:D_hybrid}, it is seen that the distortion of a heterogeneous sensing system is much larger when a single coding scheme is used (i.e., the coded scheme or the uncoded scheme) than that when the hybrid coding is used.
    It is also observed that for each given number $K$ of nodes, the distortions achieved by the proposed schemes (e.g., the sorted greedy, the pure greedy, and the group greedy algorithms) closely approach that of global searching, i.e., we have $\widetilde{D} =1$ for each of them.

Compared with the pure greedy algorithm (the curve with $\times$), we observe in Figs. \ref{fig:d_ratio_vs_K} and \ref{fig:policy_error_ratio_vs_K} that the sorted greedy algorithm (the curve with $\circ$) performs better in terms of normalized distortion when $K$ is small and performs slightly worse when $K$ gets larger;
     the policy error rate of the sorted greedy algorithm is always smaller than that of the pure greedy algorithm.
Since the computational complexity of the sorted greedy algorithm is much lower, however,
    the sorted greedy algorithm should be a simple yet powerful algorithm.
For the group greedy algorithm (labeled by $\triangledown$, {\small$\square$}, or $\diamond$), we see that its normalized distortion and policy error rate coincide with that of the pure greedy algorithm if $L=1$.
    When group size $L$ is large (e.g., $L=10$ and $L=32$), the normalized distortion decreases to unity and the policy error rate is also reduced, as shown by curves labeled by {\small$\square$} and $\diamond$.

The effectiveness of increasing group size $L$ is illustrated in Figs. \ref{fig:d_ratio_vs_L} and \ref{fig:policy_error_ratio_vs_L}, in which the numbers of simulation rounds and the number of nodes are set to $N_{\text{sim}}=5000$ and $K=10$, respectively.
    It is seen that when $L$ is increased, both the normalized distortion and the policy error rate decrease quickly.
Moreover, we would like to mention that the estimation distortions of all of the three proposed greedy algorithms approach that of global searching very well (with errors less than $0.17\%$ for all cases), as shown in Figs. \ref{fig:d_ratio_vs_K}  and \ref{fig:d_ratio_vs_L}.

\begin{figure}[!t]
\centering
\includegraphics[width=3.7in]{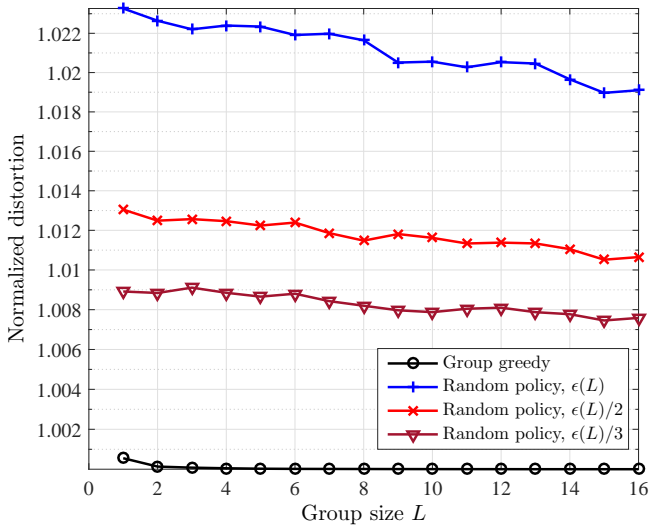}
\caption{Normalized distortion of the group greedy algorithm and policies with random errors, in which $K=10$ and $\epsilon(L)$ is the error rate of the group greedy algorithm with group size $L$.} \label{fig:random_error_policy}
\end{figure}

It is also noted in Fig. \ref{fig:policy_error_ratio_vs_L} that the policy error rate is quite large, i.e., the coding policies obtained by the proposed algorithms are different from the policy obtained by global searching.
    Meanwhile, Fig. \ref{fig:d_ratio_vs_L} shows that the distortions of the proposed algorithms closely approach that of global searching.
    However, this does not mean that optimizing the coding policy for the system is not necessary.
To show this, we investigate the normalized distortions of the group greedy algorithm and three random policies in Fig. \ref{fig:random_error_policy}.
     Specifically, for each $L\geq1$, we denote the policy error rate of the group greedy algorithm as $\epsilon(L)$ and randomly generate three coding policies, in which the coding scheme of each node is different from that of global searching with probability $\epsilon(L), \epsilon(L)/2$, and $\epsilon(L)/3$, respectively.
With the same or even smaller policy error rate, it is observed that policies with random errors yield much larger estimation
distortions, which validates the effectiveness of the proposed greedy algorithms.
    The reason is that in the proposed algorithms, the coding policies are obtained according to a certain optimizing rule and errors only occur to the nodes whose coding scheme do not affect the system distortion much.
More specifically, \textit{Theorem} \ref{th:distortion_hybrid} shows that in the distortion of the system, the contribution from either the coded node set or the uncoded node set is dominated by their respective capable nodes, i.e., the nodes with the strongest channel-observation SNR pairs.
    On the contrary, for those nodes whose channel and/or observation SNR(s) is/are small (e.g., approaches zero), their contributions to the sums in \eqref{rt:distortion_hetero} are negligible.
For these nodes, therefore, whether the obtained coding schemes are consistent with the optimal policy does not change the estimation distortion much.
If a coding policy is randomly chosen, however, the corresponding distortion would be large, especially when the coding schemes of those capable nodes are improperly set.

\section{Conclusion}\label{sec:5_conclusion}
In this paper, we investigated the estimation distortion of distributed sensing systems with  separated source-channel coding or/and  joint source-channel coding, which are implemented by separate lossy-source/channel coding and uncoded-forwarding, respectively.
    We show that the estimation distortions of the two coding schemes are closely related to the number $K$ of nodes, the observation SNRs, and the channel SNRs of the system.
Specifically, when $K$ is small or the observation SNRs are small, the coded scheme ensures smaller estimation distortion; when $K$ is large and the observation SNRs are also large, the uncoded scheme yields better estimations.
    On one hand,  the coded scheme can regulate the observation noise and the channel noise naturally, and thus reduces the estimation distortion effectively.
On the other hand, the unavoidable correlation among quantization noises prohibits the coded scheme to perform very well when $K$ is large.
    Therefore, it is reasonable to use the coded scheme and the uncoded scheme flexibly in distributed sensing systems, i.e., using the hybrid coding scheme.
In this regard, the proposed three algorithms have provided promising solutions for the optimal design of distributed sensing systems.
    Furthermore, we note that in many large scale sensing systems, it is more practical to transmit the observations with random access, which inevitably decreases the timeliness of sensing.
In our future work, therefore, we shall study the estimation distortion of distributed sensing systems with random access and practical timeliness constraints.

\appendix

\renewcommand{\theequation}{\thesection.\arabic{equation}}
\newcounter{mytempthcnt}
\setcounter{mytempthcnt}{\value{theorem}}
\setcounter{theorem}{2}
\addcontentsline{toc}{section}{Appendices}\markboth{APPENDICES}{}

\subsection{Proof of \textit{Proposition} \ref{prop:sigma_n}}  \label{prf:prop_sigma_n}
Before proving \textit{Proposition} \ref{prop:sigma_n}, we shall provide a useful lemma first.
\begin{lemma} \label{lem:corre_nqu}
    For the quantization noise $\boldsymbol{n}_{\text{qu}k}$ (cf. \eqref{df:n_quk}), its covariances with the source signal and the observation noise are, respectively, given by,
    \begin{align}
        \mathbb{E}[n_{\text{qu}k}n_{\text{ob}k}] &= \frac{ \sigma^2_{\text{ob}k} \sigma^2_{\text{qu}k} } {\sigma^2_{\text{ob}k} + \sigma^2_\theta}, \\
        \mathbb{E}[n_{\text{qu}k}\theta] &= \frac{ \sigma^2_\theta \sigma^2_{\text{qu}k} } {\sigma^2_{\text{ob}k} + \sigma^2_\theta}.
    \end{align}
\end{lemma}

\begin{figure}[!t]
\centering
\includegraphics[width=1.6in]{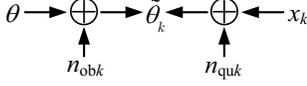}
\caption{Elements of noisy observation $\widetilde{\theta}_k$.} \label{fig:apx_model}
\end{figure}

\begin{proof}
    Note that the noisy version source signal of node $k$ is given by $\widetilde{{\theta}}_k=\theta+{n}_{\text{ob}k}$.
        Note also that the signal recovered by the fusion center can be determined by the test channel $\widetilde{\theta}_k = x_k + n_{\text{qu}k}$ (cf. \eqref{eq:x_first}), as shown in Fig. \ref{fig:apx_model}.
For any given $\widetilde{\theta}_k$, therefore, $\theta$ and ${n}_{\text{ob}k}$ would be conditionally independent of $x_k$ and ${n}_{\text{qu}k}$.

Consider a random vector given by $\boldsymbol{\xi}=[{n}_{\text{ob}k}, {n}_{\text{qu}k}, \widetilde{\theta}_k]^{\text{T}}$, for which the  covariance matrix and the precision matrix  are, respectively, given by
    \begin{align}\label{apx_E_2}
    \boldsymbol{\Sigma}_{\xi} &=
    \left[
      \begin{array}{ccc}   %该矩阵一共3列，每一列都居中放置
        \sigma^2_{\text{ob}k} & c & \sigma^2_{\text{ob}k} \\  %第一行元素
        c  & \sigma^2_{\text{qu}k} & \sigma^2_{\text{qu}k} \\  %第二行元素
        \sigma^2_{\text{ob}k} & \sigma^2_{\text{qu}k} & \sigma^2_{\text{ob}k}+\sigma^2_\theta
      \end{array}
    \right], \\
    \textbf{Q}_{\xi} &=
    \left[
      \begin{array}{ccc}   %该矩阵一共3列，每一列都居中放置
        q_{11} & 0 & q_1 \\  %第一行元素
        0  & q_{22} & q_2 \\  %第二行元素
        q_1 & q_2 & q_{33}
      \end{array}
    \right],
\end{align}
in which $c=\mathbb{E}[n_{\text{qu}k}n_{\text{ob}k}] $ and $q_{ij}=(\textbf{Q}_{\xi})_{ij}$.
     By \textit{Lemma} \ref{lem:q_cond_0} and the fact that $n_{\text{ob}k}$ and $n_{\text{qu}k}$ are conditionally independent for any given $\widetilde{\theta}_k$, we have $q_{12}=q_{21}=0$.

Since the precision matrix is the inverse of the covariance matrix, we know that $\textbf{M}_{\xi}=\boldsymbol{\Sigma}_{\xi} \textbf{Q}_{\xi}$ would be a unit matrix.
    By using $(\textbf{M}_{\xi})_{11}=1$, $(\textbf{M}_{\xi})_{31}=0$, and $(\textbf{M}_{\xi})_{21}=0$, we have the following equations,
    \begin{align} \label{apx:lem_1}
        \sigma^2_{\text{ob}k} q_{11} + \sigma^2_{\text{ob}k} q_1 &=1, \\ \label{apx:lem_2}
        \sigma^2_{\text{ob}k} q_{11} + (\sigma^2_{\text{ob}k}+\sigma^2_\theta) q_1 &=0, \\ \label{apx:lem_3}
        c q_{11} + \sigma^2_{\text{qu}k} q_1 &=0,
    \end{align}
    from which we have $\mathbb{E}[n_{\text{qu}k}n_{\text{ob}k}] =c=  \sigma^2_{\text{ob}k} \sigma^2_{\text{qu}k} / (\sigma^2_{\text{ob}k} + \sigma^2_\theta)$.
        %Note that all other unknown parameters in the precision matrix $\textbf{Q}_{\xi}$ can also be solved.

    Likewise, by considering the covariance matrix and the precision matrix of random vector $\boldsymbol{\eta}=[\theta, {n}_{\text{qu}k}, \widetilde{\theta}_k]^{\text{T}}$, we have
    $\mathbb{E}[n_{\text{qu}k}\theta] = \sigma^2_\theta \sigma^2_{\text{qu}k} / (\sigma^2_{\text{ob}k} + \sigma^2_\theta)$. This completes the proof of \textit{Lemma} \ref{lem:corre_nqu}.
\end{proof}

With \textit{Lemma} \ref{lem:corre_nqu}, \textit{Proposition} \ref{prop:sigma_n} can be proved as follows.
\begin{proof}
    For notational simplicity, we rewrite the elements of $\boldsymbol{\Sigma}_{\widetilde{\boldsymbol{n}}}$ as $(\boldsymbol{\Sigma}_{\widetilde{\boldsymbol{n}}})_{kj}=c_{kj}$ for all $k,j\in\mathcal{K}$,
         denote $(\boldsymbol{\Sigma}_{\widetilde{\boldsymbol{n}}})_{k,K+1} = (\boldsymbol{\Sigma}_{\widetilde{\boldsymbol{n}}})_{K+1,k}$ as $c_k$ for  $1\leq k\leq K+1$.
    Also, we denote $(\textbf{Q}_{\widetilde{\boldsymbol{n}}})_{kj}$ as $q_{kj}$ for all $k\in\mathcal{K}$ and denote
    $(\textbf{Q}_{\widetilde{\boldsymbol{n}}})_{K+1,k} = (\textbf{Q}_{\widetilde{\boldsymbol{n}}})_{k, K+1}$ as $q_k$ for   $1\leq k\leq K+1$.

    It is clear that $c_{K+1}=\mathbb{E}[\theta^2]=\sigma^2_\theta$ is the source signal power and $c_{kk}=\mathbb{E}[n_k^2]$ is the noise power at node $k$.
        For each $k\in\mathcal{K}$, we then have
    \begin{align}
        c_{kk} & = \mathbb{E}[(n_{\text{qu}k}-n_{\text{ob}k})^2] \\
      %\nonumber
%       & =\mathbb{E}[n_{\text{qu}k}^2] + \mathbb{E}[n_{\text{ob}k}^2]
%                - 2\mathbb{E}[n_{\text{qu}k} n_{\text{ob}k}] \\
       \label{apx:a_c_ii}
       &=\sigma^2_{\text{ob}k} + \sigma^2_{\text{qu}k} -
                    \frac{2\sigma^2_{\text{ob}k} \sigma^2_{\text{qu}k}}
                    {\sigma^2_\theta+\sigma^2_{\text{ob}k}},
    \end{align}
    in which \eqref{apx:a_c_ii} follows \textit{Lemma} \ref{lem:corre_nqu}.

    Likewise, the correlation between $n_k$ and $\theta$ is given by
 \begin{equation}\label{apx:a_ci}
          c_{k}  = \mathbb{E}[(n_{\text{qu}k}-n_{\text{ob}k})\theta] = \mathbb{E}[n_{\text{qu}k} \theta]
           = \frac{\sigma^2_\theta \sigma^2_{\text{qu}k}}
                    {\sigma^2_\theta+\sigma^2_{\text{ob}k}}.
        \end{equation}

%Note that the correlation matrix of $\widetilde{\boldsymbol{n}}$ can be expressed as $\boldsymbol{\Sigma}_{\widetilde{\boldsymbol{n}}}=\mathbb{E}[\widetilde{\boldsymbol{n}} \widetilde{\boldsymbol{n}}^{\text{T}}]$.
    For any non-zero vector $\boldsymbol{v}$,  we have $\boldsymbol{v}^{\text{T}} \boldsymbol{\Sigma}_{\widetilde{\boldsymbol{n}}} \boldsymbol{v}
    =\boldsymbol{v}^{\text{T}} \mathbb{E}[\widetilde{\boldsymbol{n}} \widetilde{\boldsymbol{n}}^{\text{T}}] \boldsymbol{v}
    = \mathbb{E}[(\boldsymbol{v}^{\text{T}}\widetilde{\boldsymbol{n}}) (\boldsymbol{v}^{\text{T}}\widetilde{\boldsymbol{n}})^{\text{T}}]
    = \mathbb{E}[|\boldsymbol{v}^{\text{T}}\widetilde{\boldsymbol{n}}|^2] \geq0$,
which means that $\boldsymbol{\Sigma}_{\widetilde{\boldsymbol{n}}}$ is non-negative definitive.
    Moreover, the independence among observation noises $n_{\text{ob}k}$ and source signal $\theta$ implies that neither  element of vector ${\widetilde{\boldsymbol{n}}}$ can be expressed as a linear combination of other elements.
Thus, $\boldsymbol{\Sigma}_{\widetilde{\boldsymbol{n}}}$ would be strictly positive definitive.
    According to \textit{Lemma} \ref{lem:q_cond_0} and the fact that $n_{\text{qu}k}$ and $n_{\text{qu}j}$ are conditionally independent when  $\theta$ is given (cf. \eqref{rt:cond_independ}), we then have $q_{kj}=0$.

 Next, we shall solve the remaining unknown elements in $\boldsymbol{\Sigma}_{\widetilde{\boldsymbol{n}}}$ and $\textbf{Q}_{\widetilde{\boldsymbol{n}}}$ from $\boldsymbol{\Sigma}_{\widetilde{\boldsymbol{n}}}^\text{T}\textbf{Q}_{\widetilde{\boldsymbol{n}}} =\textbf{I}_{K+1}$ and $\textbf{Q}_{\widetilde{\boldsymbol{n}}}^\text{T} =\textbf{Q}_{\widetilde{\boldsymbol{n}}}$, which are equivalent to
 \begin{align}
  c_{kk}q_{kk} +c_{k}q_{k} &= 1, \\
  ~~~ \sum\nolimits_{j=1}^K c_{k}q_{k} + \sigma^2_\theta q_{K+1} & = 1, \\
   c_{kj} q_{jj} +c_k q_j &=0,\\
   c_{j} q_{jj} +\sigma^2_\theta q_j &=0,
 \end{align}
for all $ k\neq j$ and $k,j\in \mathcal{K}$.

With some mathematical manipulations, we have
\begin{equation}
    c_{kj} = \frac{\sigma^2_\theta\sigma^2_{\text{qu}k} \sigma^2_{\text{qu}j}} {(\sigma^2_\theta+\sigma^2_{\text{ob}k})(\sigma^2_\theta+\sigma^2_{\text{ob}j})}.
\end{equation}
Thus, \textit{Proposition} \ref{prop:sigma_n} is proved.

\end{proof}

\subsection{Proof of \textit{Theorem} \ref{th:distortion}} \label{prf:distortion}
\begin{proof}
    Based on \textit{Proposition} \ref{prop:sigma_n}, we rewrite the covariance matrix $\boldsymbol{\Sigma}_{\boldsymbol{n}}$ as
    \begin{equation}
        \boldsymbol{\Sigma}_{\boldsymbol{n}} = \boldsymbol{\Lambda} +b\boldsymbol{u} \boldsymbol{v}^{\text{T}},
    \end{equation}
    in which
    \begin{align}
        \boldsymbol{\Lambda}&=\text{diag}(\lambda'_1, \cdots,\lambda'_K), \\
        \boldsymbol{u} &= \boldsymbol{v} = (u_1,\cdots,u_K)^{\text{T}},\\
        \lambda'_k & = \sigma^2_{\text{ob}k}  + \frac{ \sigma^2_{\text{qu}k}}
                         {(\sigma^2_\theta + \sigma^2_{\text{ob}k})^2}(\sigma^4_\theta - \sigma^4_{\text{ob}k}-\sigma^2_\theta \sigma^2_{\text{qu}k})\\
                  & = \frac{\sigma^2_\theta\gamma_{\text{ch}k}(1+\gamma_{\text{ch}k} + \gamma_{\text{ob}k})}
                  {(1+\gamma_{\text{ch}k})^2 \gamma_{\text{ob}k} }, \\
        b&=\sigma^2_\theta.
    \end{align}

    By using the following equation \cite[Chap.~1.7, (1.7.12)]{Xianda-Matrix-2005}
    \begin{equation}
        (\textbf{A}+b\boldsymbol{u} \boldsymbol{v}^{\text{T}} )^{-1} = \textbf{A}^{-1} - \frac{b}{1+ b \boldsymbol{v}^{\text{T}}  \textbf{A}^{-1}\boldsymbol{u} } \textbf{A}^{-1} \boldsymbol{u} \boldsymbol{v}^{\text{T}} \textbf{A}^{-1},
    \end{equation}
    we then have
    \begin{align}\label{apx:sigma_11}
        \boldsymbol{\Sigma}_{\boldsymbol{n}}^{-1} =& \boldsymbol{\Lambda}^{-1} - \frac{b}{1+ b \boldsymbol{v}^{\text{T}}  \boldsymbol{\Lambda}^{-1}\boldsymbol{u} } \boldsymbol{\Lambda}^{-1} \boldsymbol{u} \boldsymbol{v}^{\text{T}} \boldsymbol{\Lambda}^{-1} \\
        \nonumber
        =& \text{diag}\left(\frac{1}{\lambda'_1}, \cdots, \frac{1}{\lambda'_K}\right)
        - \frac{\sigma^2_\theta}{1+ \sigma^2_\theta \sum_{k=1}^K \frac{u_k^2} {\lambda'_k} }\\
        \label{apx:sigma_13}
            & \qquad \cdot \left(
             \begin{array}{cccc}   %该矩阵一共3列，每一列都居中放置
                \frac{u_1^2}{{\lambda'_1}^2} & \frac{u_1}{\lambda'_1}\frac{u_2}{\lambda'_2} & \cdots & \frac{u_1}{\lambda'_1}\frac{u_K}{\lambda'_K} \\  %第一行元素
                \frac{u_2}{\lambda'_2}\frac{u_1}{\lambda'_1} & \frac{u_2^2}{{\lambda'_2}^2} & \cdots & \frac{u_2}{\lambda'_2}\frac{u_K}{\lambda'_K} \\  %第一行元素
                \vdots&\ddots&\cdots&\vdots\\
                \frac{u_K}{\lambda'_K}\frac{u_1}{\lambda'_1} &\frac{u_K}{\lambda'_K}\frac{u_2}{\lambda'_2}  & \cdots & \frac{u_K^2} {{\lambda'_K}^2} %第一行元素
            \end{array}
        \right).
    \end{align}

Together with \eqref{eq:mmse_blue}, the minimum achievable distortion can be obtained as
\begin{align}
    D = & (\textbf{1}_K^\text{T} \boldsymbol{\Sigma}_{\boldsymbol{n}}^{-1} \textbf{1}_K)^{-1}\\
        = & \sum_{k=1}^K \frac{1}{\lambda'_k} - \frac{\sigma^2_\theta}{1+ \sigma^2_\theta \sum_{k=1}^K \frac{u_k^2}{ \lambda'_k} }
               \sum_{k=1}^K \sum_{j=1}^K \left(\frac{u_k}{\lambda'_k}\frac{u_j}{\lambda'_j} \right) \\
        = & \sigma^2_\theta\left(\sum_{k=1}^K \frac{1}{\lambda_k} - \frac{1}{1+ \sum_{k=1}^K \frac{u_k^2}{ \lambda_k} }
                \left(\sum_{k=1}^K\frac{u_k}{\lambda_k} \right)^2\right)^{-1},
\end{align}
where $\lambda= \lambda'/\sigma^2_\theta$.
This completes the proof of \textit{Theorem} \ref{th:distortion}.
\end{proof}

%%对角阵
%\left(\begin{array}{cccc}
%            \frac{1}{\lambda_1} &  & \multicolumn{2}{c}{\raisebox{-0.3ex}[0pt]{\Large0}}  \\
%                & \frac{2}{\lambda_1}& &\\
%                & &\ddots &\\
%            \multicolumn{2}{c}{\raisebox{1.3ex}[0pt]{\Large0}} & & \frac{1}{\lambda_K}
%\end{array}\right)

\subsection{Proof of \textit{Theorem} \ref{th:homo_digital}} \label{prf:homo_digital}
\begin{proof}
In a homogenous sensing system, $u_k$ and $\lambda_k$ (cf. \eqref{rt:u_k}, \eqref{rt:lamda_k}) are equal for each node.
    That is,
    \begin{equation}
        u = \frac{1}{1+\gamma_{\text{ch}}} ~\text{and}~
        \lambda  = \frac{\gamma_{\text{ch}}(1+\gamma_{\text{ch}} + \gamma_{\text{ob}})}
                  {(1+\gamma_{\text{ch}})^2 \gamma_{\text{ob}} }.
    \end{equation}

    Based on \textit{Theorem} \ref{th:distortion}, we then have
\begin{align}
    \sigma^2_\theta D_\text{Coded}^{-1}&= \sum_{k=1}^K \frac{1}{\lambda_k} - \frac{1}
                {1+  \sum_{k=1}^K \frac{u_k^2}{ \lambda_k} }
                \left(\sum_{k=1}^K\frac{u_k}{\lambda_k} \right)^2 \\
              &= \frac{K}{\lambda} - \frac{1}{1+ K \frac{u^2}{ \lambda} } \cdot
                \frac{K^2u^2}{\lambda^2}  \\
              &=\frac{K}{\lambda + K  u^2}.
\end{align}

Thus, the minimum achievable distortion is given by
\begin{align}
    D_\text{Coded} = &\sigma^2_\theta\left(\frac{\lambda}{K} +  u^2 \right) \\
            =&\frac{\sigma^2_\theta}{K}
             \left( \frac{\gamma_{\text{ch}}}  {(1+\gamma_{\text{ch}}) \gamma_{\text{ob}} }
                        + \frac{K+\gamma_{\text{ch}}}  {(1+\gamma_{\text{ch}})^2}
             \right).
\end{align}

This completes the proof of \textit{Theorem} \ref{th:homo_digital}.
\end{proof}

\subsection{Proof of \textit{Theorem} \ref{th:comparison}} \label{prf:comparison}
\begin{proof}
    The difference between $D_{\text{Coded}}$ (cf. \textit{Theorem} \ref{th:homo_digital}) and $D_{\text{Uncoded}}$ (cf. \textit{Proposition} \ref{prop:homo_af}) can be expressed as
    \begin{align}
        \Delta D =& D_{\text{Coded}}-D_{\text{Uncoded}} \\
            %=& \frac{\sigma^2_\theta}{K}
%              \left( \frac{\gamma_{\text{ch}}}  {(1+\gamma_{\text{ch}}) \gamma_{\text{ob}} }
%                        + \frac{K+\gamma_{\text{ch}}}  {(1+\gamma_{\text{ch}})^2}
%             \right)\\
%            & \hspace{21.3mm}- \frac{\sigma^2_\theta}{K} \left( \frac{1}{\gamma_{\text{ob}}}
%                    + \frac{1}{\gamma_{\text{ch}}}
%                    + \frac{1}{\gamma_{\text{ob}}\gamma_{\text{ch}}}
%                    \right) \\
   %         =& \frac{\sigma^2_\theta}{K} \left(
%                    \frac{1-\frac{1}{\gamma_{\text{ob}}}} {1+\gamma_{\text{ch}}}
%                    + \frac{K-1} {(1+\gamma_{\text{ch}})^2}
%                    - \frac{1}{\gamma_{\text{ch}}}
%                    - \frac{1}{\gamma_{\text{ob}}\gamma_{\text{ch}}}
%                    \right) \\
                \nonumber
            =& \frac{\sigma^2_\theta}{K\gamma_{\text{ch}} (1+\gamma_{\text{ch}})^2 }
                    \Big ( (K-2)\gamma_{\text{ch}} - 1 \\
                    \label{apx:delta_D}
            &    \hspace{25.3mm} - \frac{1}{\gamma_{\text{ob}}} (\gamma_{\text{ch}}+1) (2\gamma_{\text{ch}}+1)
                    \Big).
    \end{align}

    For $K=1$ and $K=2$, it is clear that $\Delta D<0$, and thus the coded scheme outperforms the uncoded scheme.

    For $K\geq 3$, we note that $\Delta D<0$ is equivalent to
    \begin{equation}
        (K-2)\gamma_{\text{ch}} - 1 - \frac{1}{\gamma_{\text{ob}}} (\gamma_{\text{ch}}+1) (2\gamma_{\text{ch}}+1)<0.
    \end{equation}

    Since $\gamma_{\text{ch}}$ is positive, we then have
    \begin{equation}\label{apx:dr_77}
    \frac{(K-2)\gamma_{\text{ch}}-1}{(\gamma_{\text{ch}}+1)(2\gamma_{\text{ch}}+1)} < \frac{1}{\gamma_{\text{ob}}}.
    \end{equation}

    First, it is noted that if
    \begin{equation}\label{apx:dr_79}
    \gamma_{\text{ch}}<1/(K-2),
    \end{equation}
     is satisfied, \eqref{apx:dr_77} would be true for all $\gamma_{\text{ob}}>0$.

    Second, if $\gamma_{\text{ch}}\geq 1/(K-2)$, \eqref{apx:dr_77} can be expressed as
    \begin{equation}\label{apx:dr_78}
    \gamma_{\text{ob}} < \frac{(\gamma_{\text{ch}}+1)(2\gamma_{\text{ch}}+1)} {(K-2)\gamma_{\text{ch}}-1}.
    \end{equation}

    By combining \eqref{apx:dr_79} and \eqref{apx:dr_78}, the proof of \textit{Theorem} \ref{th:comparison} would be completed.
\end{proof}

\subsection{Proof of \textit{Theorem} \ref{th:condition_hetero}} \label{prf:condition_hetero}
\begin{proof}
    Using the result given by \eqref{eq:distortion_af_cui} and \cite{Cui-Estimation-2007}, the distortion of an uncoded heterogeneous sensing system can be obtained as
    \begin{equation} \label{apx:d_uncod_heto}
        D_{\text{Uncoded}}^{\text{hetero}} = \sigma^2_\theta \left( \sum_{k=1}^K \frac{1}
        { \frac{1}{\gamma_{\text{ob}k}} + \frac{1}{\gamma_{\text{ch}k}} + \frac{1} {\gamma_{\text{ob}k}\gamma_{\text{ch}k}} }
        \right)^{-1}.
    \end{equation}

    Next, we are interested in the sign of the following difference between the inverse distortions
    \begin{equation} \label{apx:diff_heter}
        \Delta_{\text{hetero}} = \sigma^2_\theta D^{-1}_{\text{Coded}} - \sigma^2_\theta D^{-1}_{\text{Uncoded}},
    \end{equation}
     where $D_{\text{Coded}}$ is given by \eqref{rt:distortion} and $D_{\text{Uncoded}}$ is given by \eqref{apx:d_uncod_heto}.

    With some mathematical manipulations, we see that $\Delta_{\text{hetero}}$ is equal to
    \begin{align}
    \nonumber
    & \frac{1} { 1+\sum_{k=1}^K
    \frac{\gamma_{\text{ob}k}}{(1+\gamma_{\text{ch}k}+\gamma_{\text{ob}k})\gamma_{\text{ch}k}} }
    \left( \sum_{k=1}^K
                    \frac{ (1+2\gamma_{\text{ch}k})\gamma_{\text{ob}k}}
                      {(1+\gamma_{\text{ch}k}+\gamma_{\text{ob}k})\gamma_{\text{ch}k} }
    \right. \\
    \nonumber
    & + \sum_{k=1}^K\frac{ (1+2\gamma_{\text{ch}k})\gamma_{\text{ob}k}}
                      {(1+\gamma_{\text{ch}k}+\gamma_{\text{ob}k})\gamma_{\text{ch}k} }
                      \sum_{k=1}^K\frac{ \gamma_{\text{ob}k}}
                      {(1+\gamma_{\text{ch}k}+\gamma_{\text{ob}k})\gamma_{\text{ch}k} }   \\
    \label{apx:equivalence}
    & \left. - \left( \sum_{k=1}^K
                            \frac{ (1+\gamma_{\text{ch}k})\gamma_{\text{ob}k}}
                            {(1+\gamma_{\text{ch}k}+\gamma_{\text{ob}k})\gamma_{\text{ch}k} }
                    \right)^2
      \right).
    \end{align}

    Note that $\Delta_{\text{hetero}}$ has the same sign with the equation within the bracket in equation \eqref{apx:equivalence}.
        With some mathematical manipulations, we finally see that $\Delta_{\text{hetero}}>0$ if
        \begin{equation}
        \sum_{k=1}^{K} \frac{ (1+2\gamma_{\text{ch}k})\gamma_{\text{ob}k}}
        {(1+\gamma_{\text{ch}k}+\gamma_{\text{ob}k})\gamma_{\text{ch}k}}
        - \left( \sum_{k=1}^{K} \frac{ \gamma_{\text{ob}k}}
        {1+\gamma_{\text{ch}k}+\gamma_{\text{ob}k}}
        \right)^2 >0,
        \end{equation}
    which completes the proof of  \textit{Theorem} \ref{th:condition_hetero}.
\end{proof}

\subsection{Proof of \textit{Theorem} \ref{th:distortion_hybrid}} \label{prf:distortion_hybrid}
\begin{proof}
    Without loss of generality, we assume that the first $K_1$ nodes use the coded scheme and the rest ones use the uncoded scheme.
        We denote the signal used to estimate $\theta$ at the fusion center as $\boldsymbol{x}=\textbf{1}_K \theta + \boldsymbol{n}_{\text{h}} $, in which $\boldsymbol{n}_{\text{h}}$ is composed of the noise $\boldsymbol{n}_\text{{c}}$ from nodes using the coded scheme and the noise $\boldsymbol{n}_\text{{u}}$ from nodes using the uncoded scheme, i.e., $\boldsymbol{n}_{\text{h}} =[\boldsymbol{n}_\text{{c}}^{\text{T}}, \boldsymbol{n}_\text{{u}}^{\text{T}} ]^{\text{T}}$.
    Since $\boldsymbol{n}_\text{{c}}$ and  $\boldsymbol{n}_\text{{u}}$ are noises of nodes from disjoint sets, they must be independent from each other.
        By using a BLUE estimator, the estimation of $\theta$ is obtained from  $\hat{\theta} = \boldsymbol{f}^\text{T} \boldsymbol{x}$ and the corresponding distortion is given by (cf. \eqref{eq:mmse_blue})
        \begin{equation}\label{apx:d_heter}
    D_{\text{Hybrid}} =  (\textbf{1}_K^\text{T} \boldsymbol{\Sigma}_{\text{h}}^{-1} \textbf{1}_K)^{-1},
\end{equation}
        in which $\boldsymbol{\Sigma}_{\text{h}}$ is the covariance of noise vector $\boldsymbol{n}_{\text{h}}$.
We further express $\boldsymbol{\Sigma}_{\text{h}}$ as
\begin{equation}
    \boldsymbol{\Sigma}_{\text{h}} =
    \left[
      \begin{array}{cc}   %该矩阵一共3列，每一列都居中放置
        \boldsymbol{\Sigma}_{\text{cc}} & \boldsymbol{\Sigma}_{\text{cu}} \\  %第一行元素
        \boldsymbol{\Sigma}_{\text{uc}} & \boldsymbol{\Sigma}_{\text{uu}} \\  %第二行元素
      \end{array}
    \right],
\end{equation}
in which $\boldsymbol{\Sigma}_{\text{cc}}$ is the covariance matrix of $\boldsymbol{n}_\text{{c}}$,
$\boldsymbol{\Sigma}_{\text{uu}}$ is the covariance matrix of $\boldsymbol{n}_\text{{u}}$,
$\boldsymbol{\Sigma}_{\text{cu}}$ and $\boldsymbol{\Sigma}_{\text{uc}}$ are cross correlation matrices between $\boldsymbol{n}_\text{{c}}$ and $\boldsymbol{n}_\text{{u}}$.

First, $\boldsymbol{\Sigma}_{\text{cc}}$ can readily be obtained from \textit{Proposition} \ref{prop:sigma_n} and its inverse matrix $\boldsymbol{\Sigma}_{\text{cc}}^{-1}$ can be obtained in a similar way as \textit{Appendix} \ref{prf:distortion} (cf. \eqref{apx:sigma_11}, \eqref{apx:sigma_13}).

    Second, since $\boldsymbol{n}_\text{{c}}$ is independent from $\boldsymbol{n}_\text{{u}}$, we have
\begin{equation}
    \boldsymbol{\Sigma}_{\text{cu}}=\textbf{0}_{K_1\times K_0} ~\text{and}~\boldsymbol{\Sigma}_{\text{uc}}=\textbf{0}_{K_0\times K_1},
\end{equation}
which means that the inverse matrix of $\boldsymbol{\Sigma}_{\text{h}} $ would be
\begin{equation}\label{apx_E_2}
    \boldsymbol{\Sigma}_{\text{h}}^{-1} =
    \left[
      \begin{array}{cc}   %该矩阵一共3列，每一列都居中放置
        \boldsymbol{\Sigma}_{\text{cc}}^{-1} & \boldsymbol{0} \\  %第一行元素
        \boldsymbol{0} & \boldsymbol{\Sigma}_{\text{uu}}^{-1} \\  %第二行元素
      \end{array}
    \right].
\end{equation}

According to \cite{Cui-Estimation-2007} and \eqref{eq:distortion_af_cui}, we further have $\boldsymbol{\Sigma}_{\text{uu}}= \text{diag}\{d_1, d_2, \cdots, d_{K_0}\}$, in which
\begin{equation}\label{apx_E_3}
        d_{k} =  \frac{1}{\gamma_{\text{ob}k}}
                    + \frac{1}{\gamma_{\text{ch}k}}
                    + \frac{1}{\gamma_{\text{ob}k}\gamma_{\text{ch}k}}.
\end{equation}

By combining \eqref{apx:d_heter}, \eqref{apx:sigma_11}, \eqref{apx:sigma_13}, \eqref{apx_E_2}, and \eqref{apx_E_3}, the proof of \textit{Theorem} \ref{th:distortion_hybrid} would be completed.

\end{proof}

%\section*{Acknowledgement}
%This work was supported in parts by the National Natural Science Foundation of China (NSFC) under Grant 61701247,  the Jiangsu Provincial Natural Science Research Project under Grant 17KJB510035,
%%the Project Funded by the Priority Academic Program Development of Jiangsu Higher Education Institutions,
%and the Startup Foundation for Introducing Talent of NUIST under Grant 2243141701008;
%    Pingyi Fan's work was supported in parts by the China Major State Basic Research Development Program (973 Program) No. 2012CB316100(2) and the National Natural Science Foundation of China (NSFC) No. 61171064 and NSFC No. 61621091.

\small{
\bibliographystyle{IEEEtran}

}

\end{document}